\documentclass{ifacconf}

\usepackage{amsmath,amssymb,amsfonts}
\usepackage[dvipsnames]{xcolor}
\usepackage{arydshln}
\usepackage[font=footnotesize,labelformat=simple]{subcaption}

\usepackage{balance}
\usepackage{url}

\usepackage{lipsum}
\usepackage{comment}
\usepackage{graphicx}
\usepackage{epstopdf}
\usepackage{textcomp}
\usepackage{float}
\usepackage{nomencl}
\usepackage{nomencl}
\usepackage{siunitx}
\usepackage{booktabs}
\usepackage{enumitem}
\usepackage{caption}
\usepackage{caption}
\usepackage{xcolor}
\usepackage{mathrsfs}
\usepackage{tikz}
\usetikzlibrary{matrix}
\usepackage{algpseudocode}
\usepackage{algorithm}
\usetikzlibrary{positioning}
\usepackage{nicematrix}
\usepackage{blkarray}
\usepackage{placeins}
\usepackage{makecell}
\usepackage{mathtools, cuted}    
\usepackage{natbib}        

\newcommand{\R}{\mathbb{R}}

\newcommand{\X}{\mathbb{X}}
\newcommand{\T}{\mathbb{T}}

\newcommand{\bmat}[1]{\begin{bmatrix} #1\end{bmatrix}}

\newcommand{\col}{\mathrm{col}}

\newcommand{\mbf}[1]{\mathbf{ #1}}

\newcommand{\fourpi}[4]{\ensuremath{\mathfrak{\P}\hspace{-0.8ex}\left[\footnotesize\begin{array}{c|c}
#1&#2\\\hline#3 & \{#4\}
\end{array}\right]}}
\newcommand{\threepi}[1]{\mathfrak{\P}_{\{#1\}}}

\definecolor{customred}{HTML}{FBE5D6} 
\definecolor{customblue}{HTML}{DAE3F3} 
\definecolor{customgreen}{HTML}{D9F2D0} 
\definecolor{errorband}{HTML}{FF8080} 

\newcommand{\blackline}{\raisebox{2pt}{\tikz{\draw[-,black!40!black,solid,line width = 0.9pt](0,0) -- (3mm,0);}}}
\newcommand{\blacklinedashed}{\raisebox{2pt}{\tikz{\draw[-,black!40!black,dashed,line width = 0.9pt](0,0) -- (3.0mm,0);}}}

\newcommand{\blueline}{\raisebox{2pt}{\tikz{\draw[-,black!40!blue,solid,line width = 0.9pt](0,0) -- (3mm,0);}}}

\def\BibTeX{{\rm B\kern-.05em{\sc i\kern-.025em b}\kern-.08em
    T\kern-.1667em\lower.7ex\hbox{E}\kern-.125emX}}
\theoremstyle{definition}

\definecolor{sensorgray}{HTML}{A6A6A6}

\begin{document}
\begin{frontmatter}

\title{A Digital Twin of Evaporative Thermo-Fluidic
Process in Fixation Unit of DoD Inkjet Printers\thanksref{footnoteinfo}} 

\thanks[footnoteinfo]{This work was supported in part by Canon Production Printing B.V.,5914 HH Venlo, The Netherlands.}

\author[First]{Samarth Toolhally} 
\author[Second]{Joeri Roelofs} 
\author[First]{Siep Weiland}
\author[First]{Amritam Das}

\address[First]{Eindhoven University of Technology, Department of Electrical Engineering, Control Systems group, P.O. Box 513, 5600 MB Eindhoven, The
Netherlands (e-mail: s.toolhally@tue.nl, s.weiland@tue.nl, am.das@tue.nl).}
\address[Second]{Canon Production Printing B.V., Van der Grintenstraat 10, 5914 HH Venlo,
The Netherlands (e-mail: 
joeri.roelofs@cpp.canon)}

\begin{abstract}            
In inkjet printing, optimal paper moisture is essential for high print quality. Commercial printers achieve this through hot-air impingement in a fixation unit, whose drying performance is crucial to overall print quality. This paper presents a modular digital twin of the fixation unit that models the thermo-fluidic drying process and adaptively monitors its spatio-temporal performance. The core novelty lies in formulating the digital twin as an infinite-dimensional state estimator that infers spatio-temporal fixation states from limited sensor data while remaining optimally robust to external disturbances. Specifically, modularity is achieved by deriving a graph-theoretic model in which each node is governed by PDEs representing the thermo-fluidic processes within individual sections of the fixation unit. Evaporation is modeled as a nonlinear boundary effect coupled with each node’s dynamics via Linear Fractional Representation. Using the Partial Integral Equation (PIE) framework, we develop a unified approach for stability, input–output analysis, numerical simulation, and rapid prototyping of the fixation process, validated with operational data from a commercial inkjet printer. Based on the validated model, an $\mathcal{H}_{\infty}$-optimal Luenberger state estimator is synthesized to estimate the fixation unit’s thermal states from available sensor data. Together, the graph-theoretic model and optimal estimator constitute the digital twin of a commercial printer’s fixation unit, enabling real-time monitoring of spatio-temporal thermal effects on paper sheets capabilities otherwise unattainable in traditional printing processes.
\end{abstract}


\end{frontmatter} 

\section{Introduction}




Over the past several decades, the printing industry has undergone a significant transformation in technology, scale, and functionality. Historically, analogue, offset, and electrostatic presses dominated the field, serving large-volume publishing, packaging, and commercial print operations. These early production printers relied heavily on mechanical processes \citep{crompton2003printing}, required extensive setup times, and achieved cost efficiency only through very high-volume runs.
The advent of digital imaging, computer-to-plate workflows, and continuous advancements in print head technology, ink formulations, and finishing systems have fundamentally reshaped this landscape. Modern production printing systems are increasingly agile, enabling shorter print runs, faster turnaround times, and greater design flexibility. At the same time, growing emphasis on sustainability and environmental compliance has driven innovation toward energy-efficient processes, reduced material waste, and the adoption of eco-friendly inks and substrates. Today, digital production printing extends beyond traditional paper media, encompassing applications in polymers, metals, textiles, wood, and even circuit boards \citep{10.1145/1064830.1064860}.From an economic perspective, the global production printer market continues to demonstrate steady growth. Although estimates vary, most market analyses project compound annual growth rates of approximately 4–7\%  over the next 5–10 years. As of 2024, the global market is valued between USD 6–9 billion, with projections suggesting an expansion to USD 12–18 billion by the early to mid-2030s \citep{BusinessResearchInsights2024}. In the context of this market growth and the push for sustainability, drop-on-demand (DoD) inkjet technology has emerged as a core component of modern digital printing systems. Unlike older printers, DoD printing is explicitly governed by the user's demand. 

\subsection{The need of fixation in  DoD inkjet printers}
A typical DoD printer comprises two primary units: the jetting unit and the fixation unit. Based on the user-demanded image, the jetting unit precisely deposits ink droplets onto the substrate through an array of nozzles, where temperature, droplet volume, and ejection velocity are tightly controlled according to the user-defined image data. The fixation unit subsequently dries the printed medium to achieve an optimal temperature-moisture balance for desired print-quality. This process typically involves a heated moving conveyor combined with hot-air impingement to facilitate moisture evaporation, as illustrated in Figure \ref{fixation process}. 
\begin{figure}[H]
    \centering
    \includegraphics[scale=0.3]{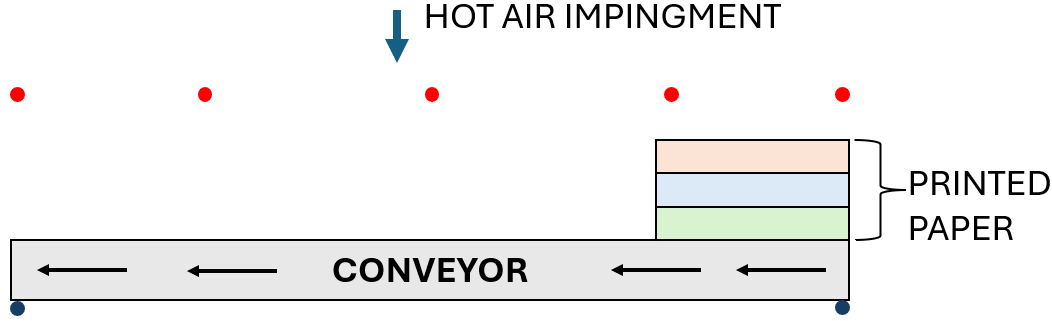}
    \caption{Fixation process. (\protect\tikz \protect\fill[blue] (0,0) circle (2pt);) represents the conveyor temperature sensors.  (\protect\tikz \protect\fill[red] (0,0) circle (2pt);) represents the paper temperature sensor.}
    \label{fixation process}
\end{figure}

\subsection{Challenges in current fixation process}
Although temperature sensors monitor the conveyor, direct measurement of the residual temperature and moisture content in the paper sheets during fixation is not possible because it is impossible for the sensors to be directly in contact with the paper sheets. Furthermore, the high throughput of commercial printing requires fast movement of the paper sheets during fixation, making alternative sensing such as thermal imaging infeasible. However, these spatio-temporal thermo-fluidic effects on the paper sheets are crucial to assess the thermo-mechanical properties of the printed paper sheets as well as the integration of ink to the medium, two key performance indices for the print quality. The state-of-the-art practice in the printing industry largely focuses on design improvements when it comes to monitor and compensate for the thermo-fluidic effects on the paper sheets during fixation, which currently poses as a major quality-limiting factor in DoD inkjet printing. A system-theoretic perspective to model and control thermo-fluidic effects is sorely missing in DoD inkjet technology \citep{9465747}. 
A limited number of efforts for model-based feedforward compensation for these thermo-fluidic effects have been made recently. However, they have been futile because of the absence of accurate spatio-temporal estimation of these effects over the course of fixation. Since Partial Differential Equations (PDEs) are used to model these thermo-fluidic phenomena \citep{pletcher2012computational}, the main computational challenge arises from their infinite dimensional nature. Standard synthesis approaches based on early lumping \citep{SHANG2000533,10.1007/978-3-0348-8849-3_12} and late lumping  suffer from the same issue: an inaccurate assessment of robustness and stability certificates, which makes them unreliable for commercial implementation. In \cite{9029595, 2024arXiv241101793B}, a state estimation strategy on infinite dimensions is provided which circumvents the issues related to lumping by converting the PDEs to an equivalent set of Partial Integral Equations (PIEs). This synthesis strategy does not involve any approximation and can be solved by convex optimization (also see \cite{PEET2019132}, \cite{sachin}). However, it is important to note that, beyond a proof of concept, the PIE framework has never been incorporated into industrial products.  

\subsection{Contribution}
\textcolor{black}{The fixation unit in commercial DoD inkjet printers must be tuned and evaluated across a wide range of operational scenarios, yet performing repeated physical experiments is both time‑consuming and costly. A physics‑based virtual replica of the fixation process enables the prediction of drying behaviour and the identification of optimal fixation parameters, thereby allowing system-level optimisation in a cost-efficient manner without the need for repeated physical experimentation. A digital twin is such a software-based replica that represents the underlying thermo-fluidic processes in the fixation unit with a multitude of functionalities that are relevant and are customizable based on application-specific requirements.}

\textcolor{black}{This work develops a modular, physics‑based digital twin for simulation, rapid prototyping, and in‑operation monitoring of the thermo‑fluidic dynamics of paper sheets during the fixation process. To support this modularity, the physical system is represented using a graph-theoretic structure in which each subsystem is modelled as a node governed by its own PDEs. The purpose of this representation is to provide a scalable organisational framework that cleanly encodes interconnections, boundary interactions, and heterogeneous material layers. This is particularly relevant in industrial printing, where paper may consist of multiple composite layers with different thermo‑fluidic characteristics. The graph-based formulation, therefore, enables straightforward adaptation to different media and fixation-unit configurations, ensuring that the digital twin remains extensible and industry relevant.} 

 To overcome the limitations of existing methods, which struggle with accurate predictions due to the inability to precisely model spatio-temporal thermo-fluidic effects, this study employs PIEs, a class of equations derived from PDEs through a variable transformation. Unlike traditional approaches, PIEs do not require boundary conditions, and their formulation allows the use of Linear Matrix Inequality methods that are widely-used for Ordinary Differential Equations. Furthermore, the digital twin incorporates an $\mathcal{H}_{\infty}$-optimal state estimation technique to monitor infinite dimensional states while accounting for the sparse sensor data and unknown external perturbations. By integrating a state estimator with a modular physical model, validated with machine data, this work provides the very first evidence that the PIE framework is implementable in an industrial setting and offers a computational advantage that is otherwise impossible to achieve as traditional computational approaches are limited by the choice of approximation scheme employed to discretize the original infinite dimensional behavior. In this way, the paper offers a pathway to exploit new computational techniques in the realm of distributed parameter systems to significantly improve the quality of DoD inkjet printing processes by more accurate modeling and estimation of the thermo-mechanical interactions during fixation.
The remainder of this paper is organized as follows: Section 2 introduces preliminaries and notation. Section 3 presents thermo-fluidic process modeling and compares simulation results with sensor data. Section \ref{estimator section} details the synthesis of the $\mathcal{H}_{\infty}$ estimator. Section 5 concludes with a summary and future work.
\section{Preliminaries}

Due to confidentiality reasons, in this paper, all the plots are normalized by subtracting the data points with a fixed nominal value. 
\textcolor{black}{Table \ref{physicalconst} shows all the physical constants and their units used in this paper.}
\begin{table}[H]
\centering
\captionsetup{labelfont={color=black}, textfont={color=black}}
\caption{\textcolor{black}{Physical constants and their units}}\label{physicalconst}
{\color{black}
\begin{tabular}{ccc}
\toprule
Symbol & Description & Units  \\
\midrule
$\kappa$ & Thermal conductivity of the material & $W/(m\cdot K)$ \\
$\rho$   & Density of the material & $kg/m^3$\\
$\chi$   & Specific heat capacity of the material & $J/(kg\cdot K)$ \\
$h$      & Heat transfer coefficient & $W/(m^2 K)$\\
$g$      & Moisture transfer coefficient & $m/s$\\
$\delta$ & Moisture diffusion coefficient & $m^2/s$\\
\bottomrule
\end{tabular}
}
\end{table}
 The space of functions $f: A\rightarrow B$ is denoted by $(B)^{A}$, i.e., $f \in (B)^{A}$. The array $\bmat{a\\b}$ is often written as $\text{col}(a, b)$. $\partial^q_s f$ denotes the q\textsuperscript{th}-order partial derivative of the function \( f \) with respect to \( s \), i.e. \( \frac{\partial^q f}{\partial s^q} \).

Given bounded matrix $P\in\R^{n\times m}$, and bounded matrix-valued polynomials $\mbf{Q_1}:[\alpha, \beta] \to  \R^{n\times p}$, $\mbf{Q_2}:[\alpha, \beta] \to  \R^{q\times m}$, $\mbf{R_0}$,$\mbf{R_1}$, and $\mbf{R_2}: [\alpha, \beta] \to \R^{q\times p} $, a Partial Integral (PI) operator  denoted by $\fourpi{P}{\mbf{Q_1}}{\mbf{Q_2}}{\mbf{R_i}}$ is defined for all $s \in [\alpha, \beta]$, $u\in \R^m$ and $\mbf v: [\alpha,\beta] \to \R^p$ according to
\begin{align}\label{eq:4pi}
\left(\fourpi{P}{\mbf{Q_1}}{\mbf{Q_2}}{\mbf{R_i}}\bmat{u\\\mbf{v}}\right)(s)
& := \bmat{Pu + \int\limits_{\alpha}^{\beta}\mbf{Q_1}(r)\mbf{v}(r)\mathrm{d}r\\\mbf{Q_2}(s)u+\left(\threepi{\mbf{R_i}}\mbf{v}\right) (s)}.
\end{align}
Here, $\threepi{\mbf{R_i}}$ is defined as follows
\begin{align}
&\left(\threepi{\mbf{R_i}}\mbf v\right)(s):=\label{eq:3pi}\\
& \mbf{R_0}(s) \mbf v(s) +\int\limits_{\alpha}^s \mbf{R_1}(s,r)\mbf v(r)\mathrm{d} r+\int\limits_s^{\beta}\mbf{R_2}(s,r)\mbf v(r)\mathrm{d} r.\notag
\end{align}


\section{Modeling the Thermo-Fluidic Process in Fixation} \label{develop digital twin}
\subsection{Graph theoretic representation}

Following \citep{phdthesis}, thermo-fluidic processes in sheets of paper is defined as a finite graph as follows.

\begin{defn}\label{def1}(Fixation Process)
The fixation process is a finite graph
\begin{align}
\label{graph}
 \mathcal{G} = (\mathbb{N}, A, \mathbb{E}),   
\end{align}

where 
\(\mathbb{N} = \{\mathcal{N}_1, \dots, \mathcal{N}_m\}\) is the set of nodes, 
\(A \in \{0,1\}^{m \times m}\) is the adjacency matrix with
\[
[A]_{i,j} = \begin{cases}
1, & \text{if $\mathcal{N}_i$ is connected to $\mathcal{N}_j$,} \\
0, & \text{otherwise,}
\end{cases}
\]
and \(\mathbb{E} = \{\mathcal{E}_{i,j} \mid A_{i,j} = 1\}\) is the set of edges.

\subsubsection{Node Structure:}

Each node \(\mathcal{N}_i \in \mathbb{N}\) is defined as
\begin{align}
\label{node}
\mathcal{N}_i = (\mathbb{X}_i, \mathbb{X}_i^{\mathrm{bc}}, \mathfrak{S}_i, \mathfrak{P}_i, \mathfrak{P}_i^{\mathrm{bc}}, \mathfrak{D}_i),
\end{align}
where:

\begin{itemize}
    \item \(\mathbb{X}_i = [s_{\iota,i}, s_{\upsilon,i}] \subset \mathbb{R}\) is the spatial domain of node \(\mathcal{N}_i\), and \(\mathbb{X}_i^{\mathrm{bc}}\) denotes its boundary exposed to external conditions.
    \item The signal space \(\mathfrak{S}_i = \mathfrak{S}_{p,i} \times \mathfrak{S}_{o,i}\) contains the internal states \(\mathfrak{S}_{p,i} = \mathbb{R}^{n_{\mathbf{x},i}}\) and external signals
    \(\mathfrak{S}_{o,i} = \mathbb{R}^{n_{d,i}+n_{w,i}+n_{y,i}+n_{z,i}+n_{l,i}+n_{p,i}+n_{q,i}}\), summarized in Table~\ref{tab:node-signals}.
    \item The node behavior is captured by a subspace
    \[
        \mathfrak{P}_i \subset (\mathfrak{S}_{p,i})^{\mathbb{X}_i \times \mathbb{T}} \times (\mathfrak{S}_{o,i})^{\mathbb{T}},
    \]
    representing the relation among all signals.
    \item Boundary behavior is specified by \(\mathfrak{P}_i^{\mathrm{bc}} \subset \mathfrak{P}_i\), restricting the node behavior at external boundaries \(\mathbb{X}_i^{\mathrm{bc}}\).
    \item Nonlinearity is modeled by introducing two latent signals $p$ and $q$ and establishing an operator $\Delta: \mathbb{R}^{n_{p,i}} \rightarrow \mathbb{R}^{n_{q,i}}$ such that
    \[
        q_i = \Delta(p_i), \quad \mathfrak{D}_i \subset (\mathbb{R}^{n_{p,i} + n_{q,i}})^\mathbb{T}.
    \]
    
\end{itemize}

\begin{table*}[h]
\centering
\caption{Node signals related to Definition \ref{def1}. Corresponding to the fixation process, physical meanings of all the variables are provided here.}
\label{tab:node-signals}
\begin{tabular}{cccc}
\toprule
Signal & Type & Domain & Description \\
\midrule
$\mathbf{x}_{i,1}$ & Temperature state ($^\circ$C) & $\mathbb{X}_i \times \mathbb{T} \rightarrow \R^{n_{\mbf x,i}}$ & Multivariate function \\
$\mathbf{x}_{i,2}$ & Moisture state ($\mathrm{g/m^2}$) & $\mathbb{X}_i \times \mathbb{T} \rightarrow \R^{n_{\mbf x,i}}$ & Multivariate function \\
$d_i$ & Constant hot air impingement input ($^\circ$C)& $\mathbb{T} \rightarrow \R^{n_{d,i}}$  & Know input \\
$w_i$ & Unknown variation in hot air impingement input ($^\circ$C)& $\mathbb{T} \rightarrow\R^{n_{w,i}}$ & Exogenous unknown input \\
$y_i$ & Thermocouple measurement ($^\circ$C)& $\mathbb{T} \rightarrow \R^{n_{y,i}}$ & Measured output \\
$z_i$ & \textcolor{black}{Estimated} output($^\circ$C) & $\mathbb{T} \rightarrow \R^{n_{z,i}}$ & Performance index \\
$p_i$ & Temperature ($^\circ$C) and moisture (kg)  & $\mathbb{T} \rightarrow \R^{n_{p,i}}$  & Node nonlinear interactions \\
$q_i$ & Evaporation enthalpy (W/s) and mass flux (kg/s) &  $\mathbb{T} \rightarrow \R^{n_{q,i}}$ & Node nonlinear interactions \\
\bottomrule
\end{tabular}
\end{table*}

\subsubsection{Edge Structure:}

Each edge \(\mathcal{E}_{i,j} \in \mathbb{E}\) is defined as
\begin{align}
\label{edge}
   \mathcal{E}_{i,j} = (\mathbb{X}_{i,j}^I, \mathfrak{C}_{i,j}^I, \mathfrak{E}_{i,j}), 
\end{align}

where:
\begin{itemize}
    \item \(\mathbb{X}_{i,j}^I \subseteq \mathbb{X}_i \cap \mathbb{X}_j\) is the interface boundary between nodes \(\mathcal{N}_i\) and \(\mathcal{N}_j\).
    \item \(\mathfrak{C}_{i,j}^I = \mathbb{R}^{n_{l_{i,j}} + n_{r_{i,j}}} \subset \mathfrak{S}_i \times \mathfrak{S}_j\) is the space of interconnection signals.
    \item The interconnection subspace \(\mathfrak{E}_{i,j} \subset \mathfrak{P}_i^{\mathrm{bc}} \times \mathfrak{P}_j^{\mathrm{bc}} \) relates input and output signals across the interface. Since not all the states from one node to the other needs to be interconnected. Thus, for interconnected signals $l_i, l_j$
\[
l_i = E_{i,j} l_j,
\]
    where \(E_{i,j}\) is a constant matrix of appropriate dimension. If nodes have unequal states, additional boundary conditions are applied to make the model consistent.
\end{itemize}
\end{defn}

\textcolor{black}{
\begin{rem}
Definition \ref{def1} provides the blueprint on how to implement the digital twin in software. In fact, following principles of object oriented programming, one may create a node class and edge class with attributes defined according to each item in the definition \eqref{node}-\eqref{edge}. For the fixation unit's model, there are six nodes and each of them are defined using six individual node objects with all the attributes specified. Similar approach is followed for declaring individual edges.
\end{rem}
}
\subsection{Graph-Theoretic Model of the Fixation Unit}
In commercial printing, the drying of paper or cardboard is governed by the simultaneous heat and moisture diffusion through composite materials \citep{holik2013handbook}. Using the graph-theoretic terminology proposed in Definition \ref{def1}, the fixation unit (c.f. Fig. \ref{fixation process}) is depicted according to Fig. \ref{gtpapernew}. 
\begin{figure}[H]
    \centering
    \includegraphics[scale=0.55]{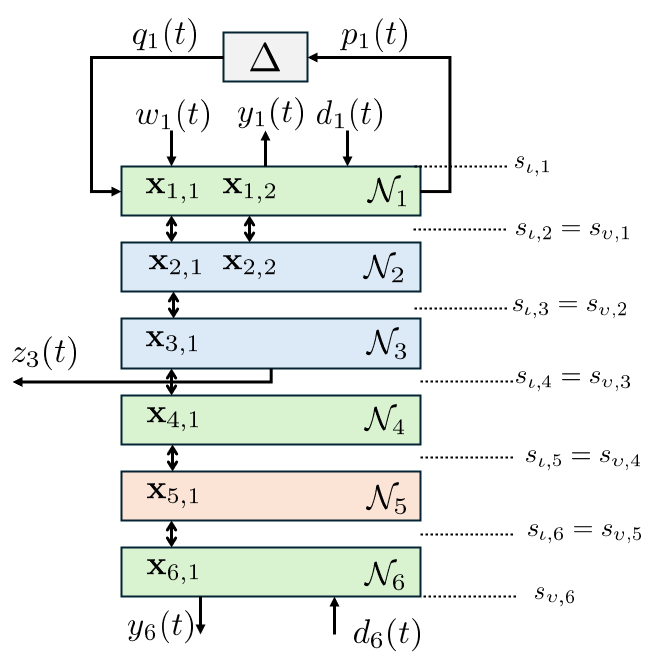}
    \caption{\textcolor{black}{Graph theoretic representation of the paper conveyor system. The nodes are color-coded where \colorbox{customblue}{\phantom{X}}, \colorbox{customred}{\phantom{X}}, and \colorbox{customgreen}{\phantom{X}} denote the paper layers, the conveyor belt, and the surrounding air layers, respectively. The arrows show the bi-directional interconnection between the same kind of states The signal $w_{1}(t)$ represents  the unknown disturbance acting on the system, while $d_{1}(t)$ and $d_{6}(t)$ are known boundary inputs. The operator $\Delta$ captures the nonlinear evaporation effects, coupling the temperature and moisture dynamics.}}
    \label{gtpapernew}
\end{figure}
The paper sheet is modeled as a non-homogeneous composite with layers $\mathcal{N}_2$ and $\mathcal{N}_3$, while node $\mathcal{N}_5$ represents the conveyor, where only temperature is defined. The dynamics at each spatially distributed node are governed by Fick's law, subject to interface and external boundary conditions. Thermal contact at interfaces may be either perfect or imperfect. Specifically, the interaction between the paper and the conveyor is modeled with imperfect thermal contact due to energy loss, whereas the thermal contact between the paper layers is assumed to be perfect. External boundaries may also experience disturbances that affect the system.

\textcolor{black}{
\begin{rem}
Only the transversal spatial coordinate is considered because longitudinal variations are significantly smaller during fixation. Moisture transport is dominated by evaporation at the exposed surfaces, making the through-thickness direction the most influential for both drying and thermal transport, while longitudinal moisture gradients remain negligible. Additionally, the conveyor belt has low thermal conductivity and a small cross-sectional area, resulting in minimal temperature variation along its length over the short fixation time. Consequently, longitudinal variations can be neglected for computational efficiency without sacrificing physical accuracy.
\end{rem}
}

Nodes are interconnected according to the adjacency matrix $A_{ij} = 1$, with edges $\mathcal{E}_{ij}$ representing state interactions through the matrices $E_{ij}$. Additional nodes $\mathcal{N}_1, \mathcal{N}_4, \mathcal{N}_6$ represent air layers above the paper, between the paper and the conveyor, and below the conveyor, respectively, as illustrated in Figure \ref{gtpapernew}. Between the two nodes associated with the paper sheet, $\mathcal{N}_2$ includes both temperature and moisture as state variables, while $\mathcal{N}_3$ only considers temperature as state variable. All the nodes are spatially distributed where $i^{\text{th}}$ node lies within the spatial domain $s_i \in [s_{\iota,i}, s_{\upsilon,i}]$ for $i \in \{1, 2, \dots, 6\}$.

The set of edges is given by $\mathcal{E}\in \{\mathcal{E}_{1,2}, \mathcal{E}_{2,1}, \mathcal{E}_{2,3},\mathcal{E}_{3,2},\mathcal{E}_{3,4},\\ \mathcal{E}_{4,3},\mathcal{E}_{4,5},\mathcal{E}_{5,4},\mathcal{E}_{5,6},\mathcal{E}_{6,5}\}$.
The corresponding adjacency matrix $A$ and the matrix $E$ defining the interconnections among all states are defined as follows, respectively. 
\begin{equation}
A = 
\begin{bmatrix}
  0 & 1 & 0 & 0 & 0 & 0\\
  0 & 0 & 1 & 0 & 0 & 0\\
  1 & 0 & 0 & 1 & 0 & 0\\
  0 & 1 & 0 & 0 & 1 & 0\\
  0 & 0 & 1 & 0 & 0 & 1\\
  0 & 0 & 0 & 1 & 0 & 0
\end{bmatrix}, E = 
\begin{bmatrix}
  0 & 0 & 1 & 0 & 0 & 0 & 0 & 0\\
  0 & 0 & 0 & 1 & 0 & 0 & 0 & 0\\
  1 & 0 & 0 & 0 & 1 & 0 & 0 & 0\\
  0 & 1 & 0 & 0 & 0 & 0 & 0 & 0\\
  0 & 0 & 1 & 0 & 0 & 1 & 0 & 0\\
  0 & 0 & 0 & 0 & 1 & 0 & 1 & 0\\
  0 & 0 & 0 & 0 & 0 & 1 & 0 & 1\\
  0 & 0 & 0 & 0 & 0 & 0 & 1 & 0
\end{bmatrix}.
\end{equation}

The dynamic equations for the system are given by 
\begin{equation}
\label{dynamicmain}
    \begin{split}   \text{$\mathcal{N}_{1}$}&:\left[\begin{array}{l}
\dot{\mathbf{x}}_{1,1}\left(s_1, t\right) \\
\dot{\mathbf{x}}_{1,2}\left(s_1, t\right)
\end{array}\right]=D_1\left[\begin{array}{l}
\partial_{s_1}^2 \mathbf{x}_{1,1}\left(s_1, t\right) \\
\partial_{s_1}^2 \mathbf{x}_{1,2}\left(s_1, t\right)
\end{array}\right] +B_{\Delta}q_1(t),\\
\text{$\mathcal{N}_{2}$}&:\left[\begin{array}{l}
\dot{\mathbf{x}}_{2,1}\left(s_2, t\right) \\
\dot{\mathbf{x}}_{2,2}\left(s_2, t\right)
\end{array}\right]=D_2\left[\begin{array}{l}
\partial_{s_2}^2 \mathbf{x}_{2,1}\left(s_2, t\right) \\
\partial_{s_2}^2 \mathbf{x}_{2,2}\left(s_2, t\right)
\end{array}\right],\\
\text{$\mathcal{N}_{3}$}&:\dot{\mathbf{x}}_{3,1}\left(s_3, t\right) = D_{3}\partial_{s_3}^2 \mathbf{x}_{3,1}\left(s_3, t\right), \\
\text{$\mathcal{N}_{4}$}&:\dot{\mathbf{x}}_{4,1}\left(s_4, t\right) = D_{4}\partial_{s_4}^2 \mathbf{x}_{4}\left(s_4, t\right), \\
\text{$\mathcal{N}_{5}$}&:\dot{\mathbf{x}}_{5,1}\left(s_5, t\right) = D_{5}\partial_{s_5}^2 \mathbf{x}_{5,1}\left(s_5, t\right), \\
\text{$\mathcal{N}_{6}$}&:\dot{\mathbf{x}}_{6,1}\left(s_6, t\right) = D_{6}\partial_{s_6}^2 \mathbf{x}_{6,1}\left(s_6, t\right).
    \end{split}
\end{equation}

Here, $D_1 = \mathrm{diag}(\frac{\kappa_{1}}{\rho_1 \chi_{1}}, \delta_1)$, $D_2 = \mathrm{diag}(\frac{\kappa_{p}}{\rho_p \chi_{p}}, \delta_p)$,  $D_{3} = \frac{\kappa_{p}}{\rho_p \chi_{p}}$, $D_{4} = \frac{\kappa_{4}}{\rho_4 \chi_{4}}$, $D_{5} = \frac{\kappa_{b}}{\rho_b \chi_{b}}$, $D_{6} = \frac{\kappa_{6}}{\rho_6 \chi_{6}}$ are the thermal diffusion coefficients given by $\frac{\kappa_p}{ \rho_p \chi_{p}}$, \\
The measured output is given by
\begin{equation}
    y_1(t) = \mathbf{x}_{1,1}(s_{\iota,1},t), \quad y_6(t)=\mathbf{x}_{6,1}({s_{\upsilon,6},t}).
\end{equation}
The observed output is given by 
\begin{equation}
    z_3(t) = \mathbf{x}_{3,1}(s_{\upsilon,3},t).
\end{equation}
\textcolor{black}{which corresponds to the temperature at the bottom surface of the paper. The estimation of $z_{3}(t)$ is chosen because it is the only thermally significant internal quantity that cannot be measured during fixation, yet it provides the missing boundary information needed to reconstruct the full through-thickness temperature profile. This profile is essential for predicting evaporation, moisture migration, and overall drying efficiency. Furthermore, after the first fixation pass, the paper is flipped for the second-side print, at which point the bottom-surface temperature becomes directly measurable. This makes $z_{3}(t)$ a physically verifiable internal variable, enabling validation of both the digital twin and the proposed state-estimation framework.}\\

The corresponding boundary conditions for all the temperature states are given by

\begin{equation}\label{BC2}
\begin{split}
        -\kappa_1 \partial_{s_1} \mathbf{x}_{1,1}(s_{\iota,1},t)+h_t(\mathbf{x}_{1,1}(s_{\iota,1},t)-T_{hai}-w_1(t))&=0,\\
        -\kappa_1\partial_{s_1} \mathbf{x}_{1,1}(s_{\upsilon,1},t)+ \kappa_p \partial_{s_2} \mathbf{x}_{2,1}(s_{\iota,2},t)&=0,\\
        -\mathbf{x}_{1,1}(s_{\upsilon,1},t)+\mathbf{x}_{2,1}(s_{\iota,2},t)&= 0,\\ 
        -\kappa_p\partial_{s_2} \mathbf{x}_{2,1}(s_{\upsilon,2},t)+ \kappa_p \partial_{s_3} \mathbf{x}_{3,1}(s_{\iota,3},t)&=0,\\
        -\mathbf{x}_{2,1}(s_{\upsilon,2},t)+\mathbf{x}_{3,1}(s_{\iota,3},t)&= 0,\\
        -\kappa_p\partial_{s_3} \mathbf{x}_{3,1}(s_{\upsilon,3},t)+ \kappa_4 \partial_{s_4} \mathbf{x}_{4,1}(s_{\iota,4},t)&=0,\\
        -\mathbf{x}_{3,1}(s_{\upsilon,3},t)+\mathbf{x}_{4,1}(s_{\iota,4},t)&= 0,\\
                  -\kappa_4\partial_{s_4} \mathbf{x}_{4,1}(s_{\upsilon,4},t)+h_{pb}(\mathbf{x}_{4,1}(s_{\upsilon,4},t)-\mathbf{x}_{5,1}(s_{\iota,5},t))&=0,\\
        -\kappa_b \partial_{s_5}\mathbf{x}_{5,1}(s_{\iota,5},t)+h_{pb}(\mathbf{x}_{4,1}(s_{\upsilon,4},t)-\mathbf{x}_{5,1}(s_{\iota,5},t))&=0,\\
        -\kappa_b\partial_{s_5} \mathbf{x}_{5,1}(s_{\upsilon,5},t)+ \kappa_6 \partial_{s_6} \mathbf{x}_{6,1}(s_{\iota,6},t)&=0,\\
        -\mathbf{x}_{5,1}(s_{\upsilon,5},t)+\mathbf{x}_{6,1}(s_{\iota,6},t)&= 0,\\
        \kappa_6 \partial_{s_6} \mathbf{x}_{6,1}(s_{\upsilon,6},t)+h_b(\mathbf{x}_{6,1}(s_{\upsilon,6},t)-T_{a})&=0.
\end{split}
\end{equation}\label{bctemp}

The boundary conditions for the moisture states are given by 
\begin{equation}\label{BC2m}
\begin{split}
        -\delta_1 \partial_{s_1} \mathbf{x}_{1,2}(s_{\iota,1},t)+g(\mathbf{x}_{1,2}(s_{\iota,1},t)-m_a)=0,\\
        -\delta_1\partial_{s_1} \mathbf{x}_{1,2}(s_{\upsilon,1},t)+ \delta_p \partial_{s_2} \mathbf{x}_{2,2}(s_{\iota,2},t)=0,\\
        -\mathbf{x}_{1,2}(s_{\upsilon,1},t)+\mathbf{x}_{2,2}(s_{\iota,2},t)= 0,\\
        \partial_{s_2} \mathbf{x}_{2,2}(s_{\upsilon,2},t) = 0.
\end{split}
\end{equation}\label{bcmoist}
\textcolor{black}{The moisture content in the printed paper disappears through evaporation in the fixation process. The mass flux of this evaporation depends on both the local paper temperature and moisture content and is captured by the static nonlinear function $\Delta_{1}(p(t))$. Likewise, the associated enthalpy of evaporation depends on the temperature and is represented by the nonlinear term $\Delta_{2}(p(t))$. Although the linear diffusion operators for temperature and moisture are block-diagonal, these nonlinear boundary terms introduce the necessary thermo–fluidic coupling. Temperature influences moisture loss through $\Delta_{1}$, and the resulting latent heat effects feed back into the heat equation through $\Delta_{2}$. Introducing coupling directly into the diffusion matrix would require non-diagonal coefficients with units inconsistent with $m^{2}/s$, rendering such a formulation physically incorrect. Modeling the interaction through $\Delta_{1}$ and $\Delta_{2}$ therefore yields a dimensionally consistent and physically accurate representation of the coupled heat–moisture dynamics.}

\textcolor{black}{Based on this coupled thermo–fluidic model, the proposed digital twin reconstructs the missing temperature and moisture profiles inside the paper. The diffusion dynamics propagate the available thermal information across the paper thickness, enabling the internal temperature distribution to be inferred even though it cannot be directly measured during fixation. This reconstructed temperature field is then used in the evaporation relations $\Delta_{1}$ and $\Delta_{2}$ to compute the local moisture loss and the associated latent-heat exchange, resulting in the full spatio-temporal moisture evolution. In this way, the digital twin provides the complete temperature and moisture fields that are otherwise unobservable, effectively functioning as a virtual sensing framework.} This is represented by the Linear Fractional Representation (LFR). The exact $\Delta_1(p(t))$ and $\Delta_2(p(t))$ functions are not mentioned as it is an intellectual property of Canon Production Printing.
\begin{figure}[h]
        \centering 
    \includegraphics[scale=0.3]{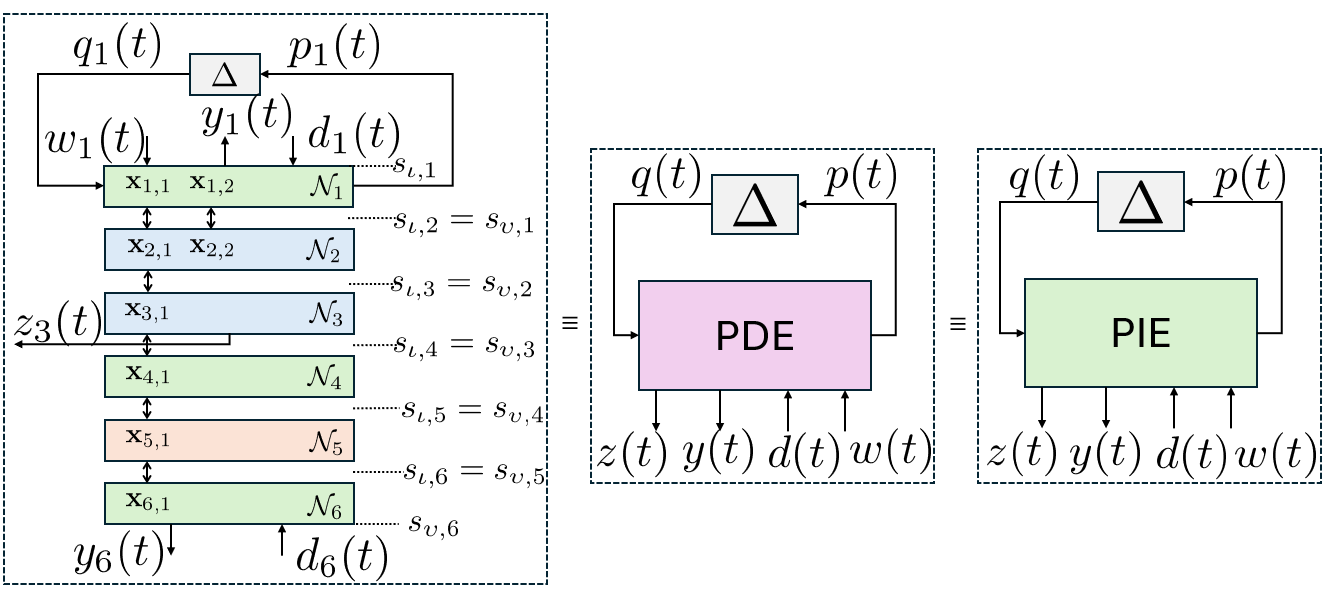}
    \caption{Equivalence between the developed model and PIEs.}
    \label{general pdeode}
\end{figure}
\begin{figure}[H]
    \centering
    \includegraphics[scale=0.40]{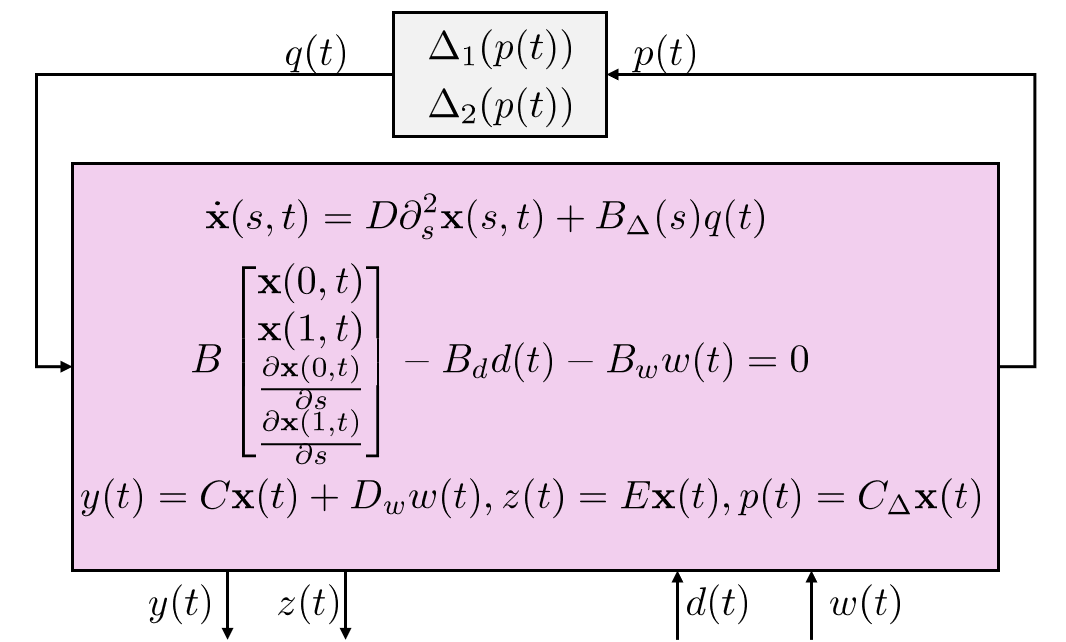}
    \caption{ Schematic of the LFR structure. $\Delta_1(p(t))$ is a static non-linear function that computes the moisture flux evaporating, and $\Delta_2(p(t))$ is a static non-linearity that computes the enthalpy required to change the state of the moisture to its vapour.}
    \label{LFRevap}
\end{figure}

\textcolor{black}{\subsection{Conversion of thermo-fluidic process model into Partial Integral Equations (PIEs)}}
\textcolor{black}{While the graph-theoretic model as well as the coupled model depicted in Fig. 4 involve boundary conditions including external inputs as well as smoothness conditions in terms of the existence of $\partial_s^2 \mathbf{x}$, as shown in \cite{sachin}, a wide class of PDE models admits a PIE representation that is free of boundary and smoothness conditions and is parametrized by PI operators. PIE representation and PI operators provide the following advantages that are suitable for the digital twin's utility:
\begin{itemize}
    \item Like matrices, composition, summation, concatenation, and adjoint  of PI operators yield a PI operator with analytical formulae for each of these operations. This property is the basis for all algorithmic analysis and synthesis procedures of a given PIE and why PIE framework is advantageous for building the digital twin.
    \item In addition to analysis and synthesis, the constraint-free PIE representation also allows for a generalized numerical scheme to approximate the bounded PI operators and simulate PDE models \citep{PEET2024115673} while the original formulation would have required model-specific numerical scheme (e.g., in compared to solving parabolic PDEs, hyperbolic PDEs often require upwind and downwind scheme). 
    \item All analysis, synthesis, and simulation methods are packaged in an open-source toolbox called \texttt{PIETOOLS} (\texttt{https://control.asu.edu/pietools/pietools},\\ also see \cite{9147712})
    where the user may define any compatible PDE model which, in turn, gets automatically converted to PIEs and any required analysis, synthesis, and simulation performed as needed on the basis of the PI operator algebra. 
\item Finally, under specific conditions, the behavior of the original PDE model is equivalent to the converted PIE model. This means any stability property \citep{PEET2021109473}, bound on $L_2$ gain between specific input-output transfer (\(\mathcal{H}_{\infty}\) norm) \citep{9030224}, performance of an optimal observer \citep{das_2019CDC} or a controller \citep{11404177} that are computed in \texttt{PIETOOLS} are equivalently valid for the original PDE model (see \cite{sachin} for more details). 
\end{itemize}}

\begin{prop} The thermo-fluidic model derived using Definition \ref{def1}, described by the dynamic equations \eqref{dynamicmain}, and boundary conditions \eqref{bctemp} and \eqref{bcmoist} is equivalent to a dynamical system $\mathfrak{P}_{\mathrm{p}}$ whose behavior is governed by Partial Integral Equations (PIE) as follows:
 \begin{equation}\begin{aligned}
			\label{def_sol_funda}
			\mathfrak{P}_{\mathrm{p}}\hspace{-0.5ex}:=\hspace{-0.5ex}\left\{\begin{array}{c}
				\mathrm{col}(w, d, p, q, z, y, \mbf v) \mid \forall t \in [0, \infty), \\
    \col\Big(p(t), q(t), z(t), y(t)\Big)\in \R^{n_p+n_q+n_z+n_y}\\
				\mathrm{col}(w(t), d(t))\in \R^{n_w+n_d},\\
    \mathrm{col}(w, d)\in C^{1}(\T, \R^{n_w+n_d})\\
				v(t) \in \mathbb{R}^{n_x} \times \prod\limits_{i} L_2^{n_{\mbf x_i}}[a,b],\\
				\hline\\ 
				 \mathscr{T} \dot{v} +  \mathscr{T}_w \dot{w} +  \mathscr{T}_d \dot{d} =\mathscr{A} v +  \mathscr{B}_1 w +  \mathscr{B}_2 d+ \mathscr{B}_q q,\\

     \bmat{z\\y\\p} = \bmat{\mathscr{C}_1&\mathscr{D}_{11}&\mathscr{D}_{12}&\mathscr{D}_{1q}\\\mathscr{C}_2&\mathscr{D}_{21}&\mathscr{D}_{22}&\mathscr{D}_{2q}\\\mathscr{C}_p&\mathscr{D}_{p1}&\mathscr{D}_{p2}&\mathscr{D}_{pq}} \bmat{v\\w\\d\\q}, \\
     q =  \Delta(p)
			\end{array}\right\}.
		\end{aligned}
        \end{equation}
Here, $\mathscr{T}, \mathscr{T}_w, \mathscr{T}_d$, $\mathscr{A}, \mathscr{B}_{11},  \mathscr{B}_{12}, \mathscr{B}_{q}$, $\mathscr{C}_1, \mathscr{D}_{11},  \mathscr{D}_{12}, \mathscr{D}_{1q}$, $\mathscr{C}_2, \\\mathscr{D}_{21},  \mathscr{D}_{22}, \mathscr{D}_{2q}$, $\mathscr{C}_p, \mathscr{D}_{p1},\mathscr{D}_{p2}, \mathscr{D}_{pq}$ are specific PI operators of appropriate dimensions.
 \end{prop} 
\begin{proof}
The process of constructing the PIEs from the graph-theoretic representation consists of three steps and is depicted in  Fig. \ref{general pdeode}.

\textit{Step 1: Constructing the LFR-PDE form with normalised domain}

We first concatenate all the node and edge relations, and, subsequently, normalise all the spatial domains between $\X = [0, 1]$. One then arrive at the following PDE coupled representation of the spatially interconnected thermo-fluidic process as depicted in Fig. \ref{LFRevap}.
Here, $\mathbf{x}:=\mathrm{col}(\mathbf{\Tilde{x}}_{j,1}, \mathbf{\Tilde{x}}_{k,2})_{j \in \mathbb{N}_{[1,6]},k \in \mathbb{N}_{[1,2]}}$ represent a set of scaled states and $D= \mathrm{diag}(\frac{D_{j,1}}{m_j^2},\frac{D_{k,2}}{m_k^2})_{j \in \mathbb{N}_{[1,6]},k \in \mathbb{N}_{[1,2]}}$. The signals $d(t)=\mathrm{col} (d_1(t), d_6(t))$ where $d_1(t) =\mathrm{col} (T_{hai}(t),m_{a}(t))$, $T_{hai}(t)$ and $m_a(t)$ are the ambient Temperature and moisture level at the boundary of the $\mathcal{N}_1$; and $d_6(t) = T_a(t)$, $T_a(t)$ is the ambient temperature at the boundary of $\mathcal{N}_6$. Furthermore, $w(t) = w_1(t)$. The matrices $B$, $B_d$ and $B_w$ are constant and related to boundary effects. Moreover, $y(t) = \mathrm{col}(y_1(t), y_6(t))$ and $z(t) = z_3(t)$.  

The signals $p(t) = p_1(t)= \mathrm{col}(\Tilde{\mathbf{x}}_{1,1}(1,t),\Tilde{\mathbf{x}}_{1,2}(1,t)$ are the temperature and moisture at $s=1$ of the node $\mathcal{N}_{1}$ and $q(t) = q_1(t)$ is a column vector of evaporation mass flux and enthalpy of evaporation based on the current temperature and moisture value. The $B_{\Delta}$ matrix converts the enthalpy into temperature.

The known initial conditions for the states are given by $\Tilde{\mathbf{x}}_{2,1}(s,0) =\Tilde{\mathbf{x}}_{3,1}(s,0)= T_{paper}$, $\Tilde{\mathbf{x}}_{2,2}(s,0)=\omega$ and $\Tilde{\mathbf{x}}_{5,1}(s,0) = T_b$. The initial conditions for the states $\Tilde{\mathbf{x}}_{1,1},\Tilde{\mathbf{x}}_{1,2},\Tilde{\mathbf{x}}_{4,1},\Tilde{\mathbf{x}}_{6,1}$ are defined such that they match the boundary conditions.\\

\textit{Step 2: Verify the invertibility of $B_T$}

For a well-defined PDE, there must be a sufficient number of boundary conditions. Without this, PDE models do not admit a PIE representation. Whether a set of boundary condition is admissible for PIE conversion or not can be easily tested by the invertibility of  the matrix $B_T$ as given by \cite[Theorem 6.2]{phdthesis} where
\begin{equation}
\label{BT matrix}
B_T= B  \begin{bmatrix}
I & I & 0 & 0 & 0 & 0 \\
0 & 0 & I & I & 0 & 0 \\
0 & 0 & 0 & I & I & I
\end{bmatrix}^\top
\end{equation}
We can indeed show that, in the case of fixation process, the matrix $B_T$ is invertible making the derived model compatible for PIE conversion.

\textit{Step 3: Apply the conversion formula from PDE to PIE}

Since the $B_T$ is invertible, using the formulae listed in \citep[Theorm 6.3]{phdthesis}, the PIE operators $\mathscr{T}, \mathscr{T}_w, \mathscr{T}_d$, $\mathscr{A}, \mathscr{B}_{11},  \mathscr{B}_{12}, \mathscr{B}_{q}$, $\mathscr{C}_1, \mathscr{D}_{11},  \mathscr{D}_{12}, \mathscr{D}_{1q}$, $\mathscr{C}_2, \mathscr{D}_{21},  \mathscr{D}_{22}, \mathscr{D}_{2q}$, $\mathscr{C}_p$, $\mathscr{D}_{p1},\mathscr{D}_{p2}, \mathscr{D}_{pq}$ are obatined.
\end{proof}



\subsection{Implementation and simulation of the model}
Based on graph-theoretic definitions, all signals associated with each node are known. The digital twin is implemented in an object-oriented framework using \texttt{MATLAB} and built on top of \texttt{PIETOOLS} version 2022. Each node is represented as a class encapsulating its properties to reduce complexity. The \texttt{layer} class defines each paper layer with parameters such as diffusion, transport, reaction, and transfer coefficients, spatial domain, boundary/initial conditions, and boundary inputs. The \texttt{spatialconveyor} class defines spatial conveyor nodes, storing parameters such as diffusion coefficient, domain size, and initial conditions. This modular design allows flexible definition of any number of states per node in arbitrary order, while ensuring correct interconnections through the $E_{i,j}$ matrix between adjacent nodes. The thermo-fluidic process is simulated using the equivalent PIE formulation from \eqref{def_sol_funda}, solved via PIE-Galerkin Projection (PGP) \citep{PEET2024115673}.

\subsection{Comparison of the fixation process model with machine data}

\textcolor{black}{This section compares simulation results with sensor measurements for two test cases, unprinted 350~$\mathrm{g/m^{2}}$ paper and printed 115~$\mathrm{g/m^{2}}$ paper. The first case involves only linear heat diffusion, while the second additionally includes moisture transport and nonlinear evaporation, providing comprehensive validation of the thermo-fluidic model. Two experimental data sets are used, one for parameter tuning based on RMSE against the mean sensor data, and another for validation. As only the paper’s top surface temperature and the conveyor’s bottom surface are measurable during operation, validation necessarily focuses on these quantities. Real-time moisture measurements are not available in practice, making temperature the only ground truth information for assessing the accuracy of the digital twin. The reliability of the reconstructed moisture profile, therefore, depends on accurately estimating the temperature distribution that drives the evaporation dynamics via the nonlinear maps $\Delta_{1}$ and $\Delta_{2}$. In this sense, the estimator acts as a virtual sensor that uses available temperature measurements to infer the internal temperature and moisture fields that cannot be directly observed.}  
\subsubsection{Simulation results for 350 $\mathrm{g/m^2}$ blank paper:}
Simulations are performed on a 350 $\mathrm{g/m^2}$ sheet of paper and are compared with sensor measurements.
Figure \ref{spatiotemporal80}.  shows the simulation result of the temperature variation of different states over their spatial domain and time. Figure \ref{sensor paper 80} and \ref{sensor conveyor  80} present the simulation results alongside the sensor readings for the top and bottom boundary of $\mathcal{N}_{2}$ and $\mathcal{N}_{5}$, respectively.  $\kappa_1$ and $\kappa_6$ are tuned to minimize the Root Mean Square Error (RMSE) and best align with the sensor readings using grid search. Table \ref{rmse80} presents the (RMSE) for the actual simulation results, along with the Normalized Root Mean Square Error (NRMSE) for the normalized data.
\begin{figure}[H]
    \centering
    \includegraphics[scale=0.3]{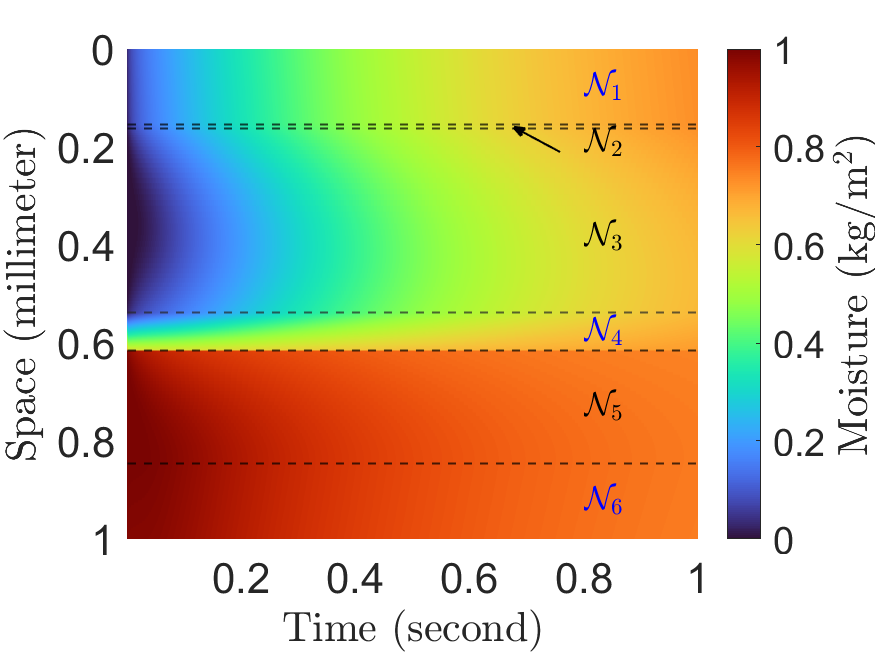}
    \caption{ Temperature variation over space and time for all the nodes of the 350 $\mathrm{g/m^2}$ paper and conveyor. (\protect\blacklinedashed) represents the boundary for each node. The black arrow shows the node with ink and moisture.}
    \label{spatiotemporal80}
\end{figure}

\begin{figure}[htbp]
    \centering
    \begin{subfigure}[b]{0.49\linewidth}
        \centering
\includegraphics[width=\linewidth]{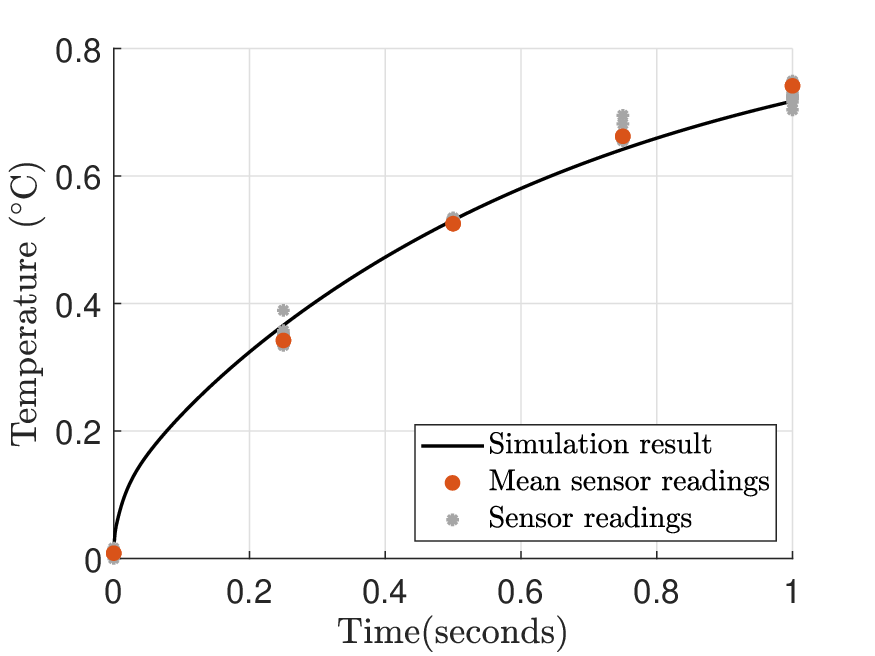}
        \caption{ (\protect\blackline) shows the temperature of the top surface of the paper obtained through simulations. (\textcolor[HTML]{A6A6A6}{*}) shows the sensor readings and (\protect\tikz \protect\fill[red] (0,0) circle (2.5pt);) shows the mean sensor reading.}
        \label{sensor paper 80}
    \end{subfigure}
    \hfill
    \begin{subfigure}[b]{0.49\linewidth}
        \centering
        \includegraphics[width=\linewidth]{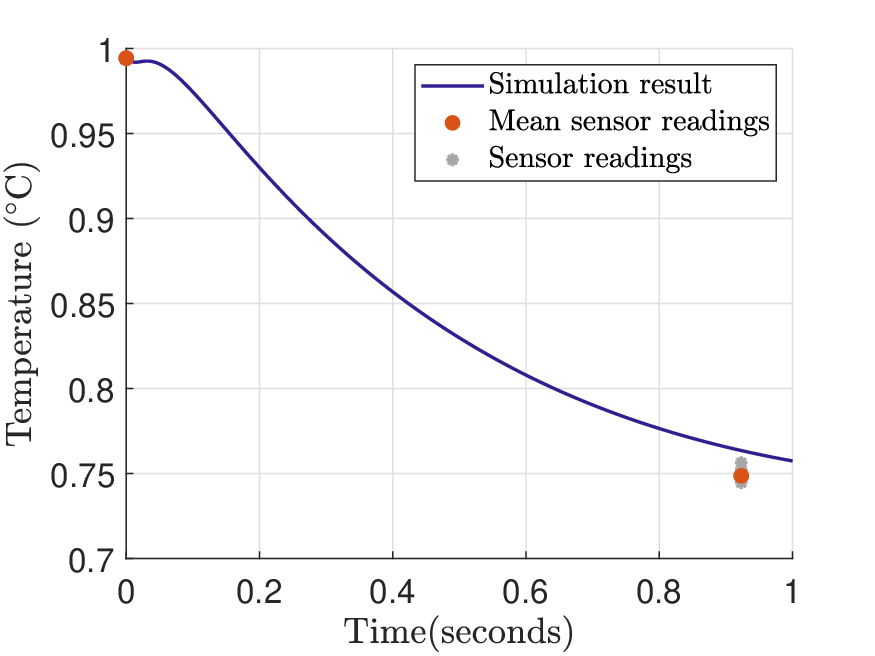}
        \caption{ (\protect\blueline) shows the temperature of the top surface of the paper obtained through simulations. (\textcolor[HTML]{A6A6A6}{*}) shows the sensor readings and (\protect\tikz \protect\fill[red] (0,0) circle (2.5pt);) shows the mean sensor reading.}
        \label{sensor conveyor  80}
    \end{subfigure}
    \caption{Comparison of simulation results and sensor readings for 350 $\mathrm{g/m^2}$ paper.}
    \label{subfigure 80gsm}
\end{figure}



\begin{table}[H]
\centering
\caption{RMSE and NRMSE for 350$g/m^2$ blank paper and converyor.}
\begin{tabular}{c c c} 
  & RMSE ($^\circ$C) & NRMSE ($^\circ$C) \\ [0.5ex] 
 \hline
  Paper & 1.91 & 0.038  \\
 Conveyor & 0.58 &  0.0096 \\ [1ex] 
 \hline
\end{tabular}
\label{rmse80}
\end{table}

\subsubsection{Simulation results for 115 $\mathrm{g/m^2}$ printed paper:}
Simulations were performed on 115 $\mathrm{g/m^2}$ paper printed with 10 $\mathrm{g/m^2}$.
Figures \ref{115temp3d} illustrates the spatial and temporal variations in temperature at all nodes. Figure. \ref{paper 115} and \ref{conveyor  115} present the simulation results alongside the sensor readings for the top and bottom boundary of $\mathcal{N}_{2}$ and $\mathcal{N}_{5}$, respectively. Figure \ref{moist 115} shows the average moisture variation in node $\mathcal{N}_{2}$. It can be observed that the variation is zero in the beginning. This is because the temperature is not high enough to change the state of the moisture to its vapor.  $\kappa_1$ and $\kappa_6$ are tuned to minimize the RMSE and best align with the sensor readings using grid search.

Table \ref{rmse115} presents the RMSE for the actual simulation results, along with the NRMSE for the normalized data.
\begin{table}[H]
\centering
\caption{RMSE and NRMSE for 115 $\mathrm{g/m^2}$ printed paper and converyor.}
\begin{tabular}{c c c} 
  & RMSE ($^\circ$C) & NRMSE ($^\circ$C) \\ [0.5ex] 
 \hline
  Paper & 2.08 &  0.0378 \\
 Conveyor & 2.65 &  0.0482 \\ [1ex] 
 \hline
\end{tabular}
\label{rmse115}
\end{table}
\begin{figure}[H]
    \centering
    \includegraphics[scale=0.3]{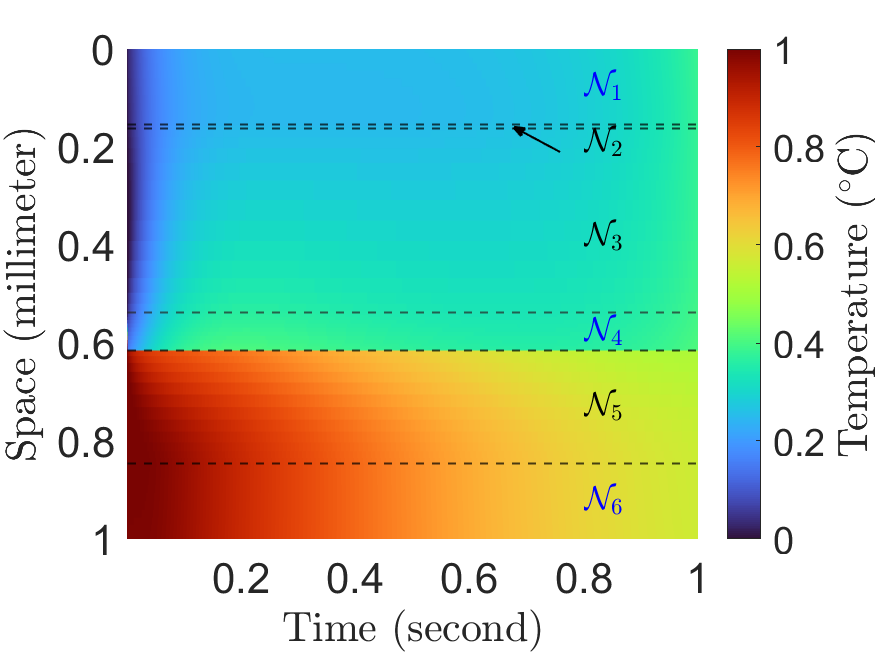}
    \caption{ Temperature variation over space and time for all the nodes of the 115 $\mathrm{g/m^2}$ paper printed with 10 $\mathrm{g/m^2}$ of ink and conveyor. 
 (\protect\blacklinedashed) represent the boundary for each node. The black arrow shows the node with ink and moisture.}
    \label{115temp3d}
\end{figure}


\begin{figure}[htbp]
    \centering
    \begin{subfigure}[b]{0.49\linewidth}
        \centering
\includegraphics[width=\linewidth]{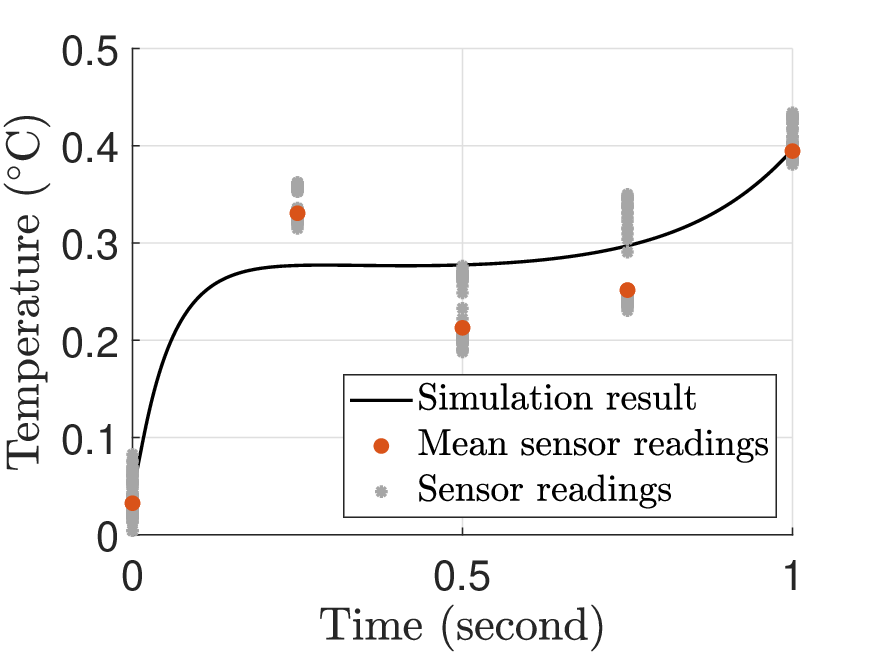}
        \caption{ (\protect\blackline) shows the temperature of the top surface of the paper obtained through simulations. (\textcolor[HTML]{A6A6A6}{*}) shows the sensor readings and (\protect\tikz \protect\fill[red] (0,0) circle (2.5pt);) shows the mean sensor reading.}
        \label{paper 115}
    \end{subfigure}
    \hfill
    \begin{subfigure}[b]{0.49\linewidth}
        \centering
        \includegraphics[width=\linewidth]{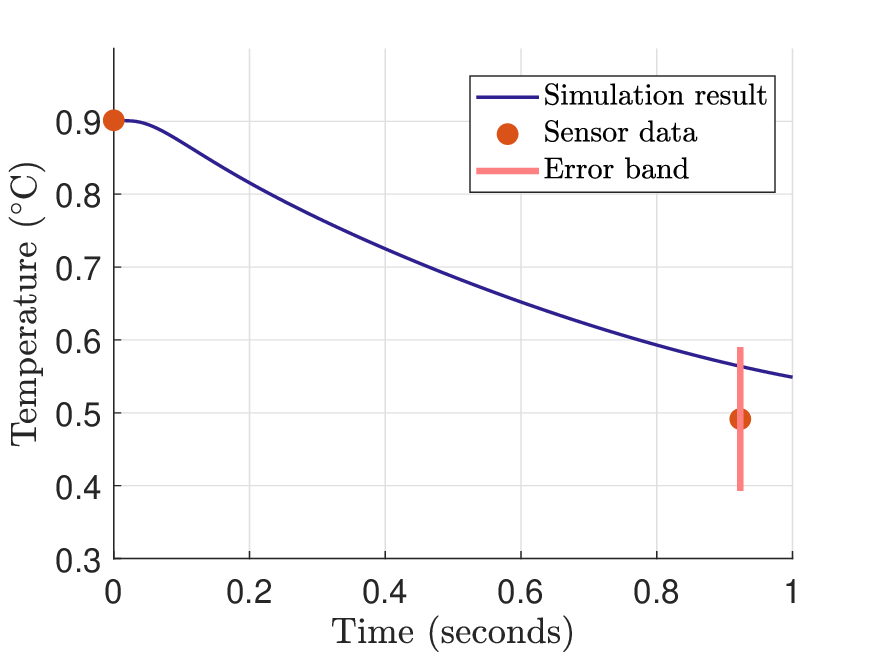}
        \caption{(\protect\blueline) shows the temperature of the top surface of the paper obtained through simulations. (\textcolor[HTML]{A6A6A6}{*}) shows the sensor readings and (\protect\tikz \protect\fill[red] (0,0) circle (2.5pt);) shows the mean sensor reading.}
        \label{conveyor  115}
    \end{subfigure}
    \caption{Comparison of simulation results and sensor readings for 115 $\mathrm{g/m^2}$ paper.}
    \label{subfigure 115gsm}
\end{figure}



\begin{figure}[htbp]
    \centering
    \begin{subfigure}[b]{0.49\linewidth}
        \centering
        \includegraphics[width=\linewidth]{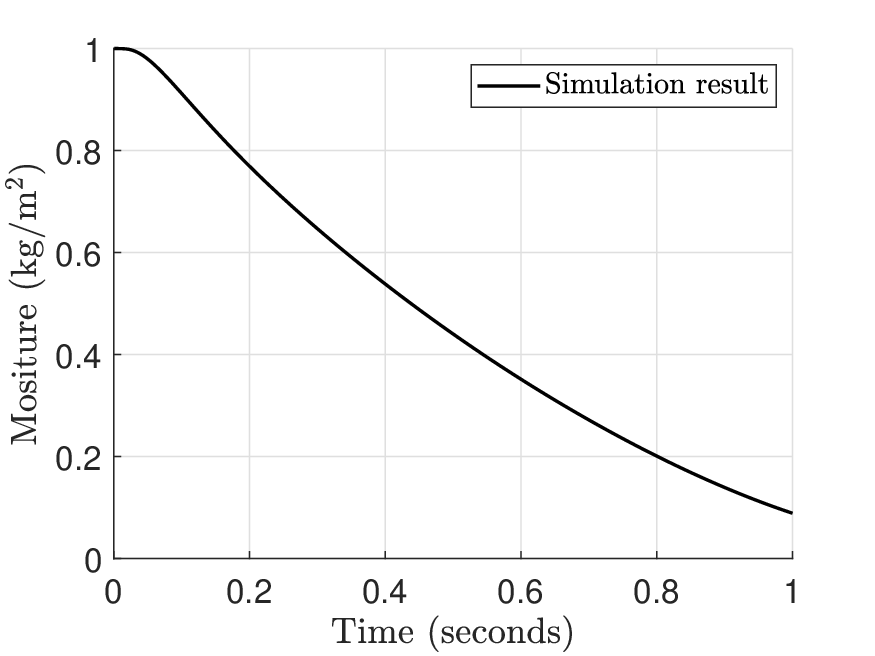}
        \caption{ (\protect\blackline) shows  average moisture variation over time in $\mathcal{N}_{2}$.}
        \label{moist 115}
    \end{subfigure}
    \hfill
    \begin{subfigure}[b]{0.49\linewidth}
        \centering
        \includegraphics[width=\linewidth]{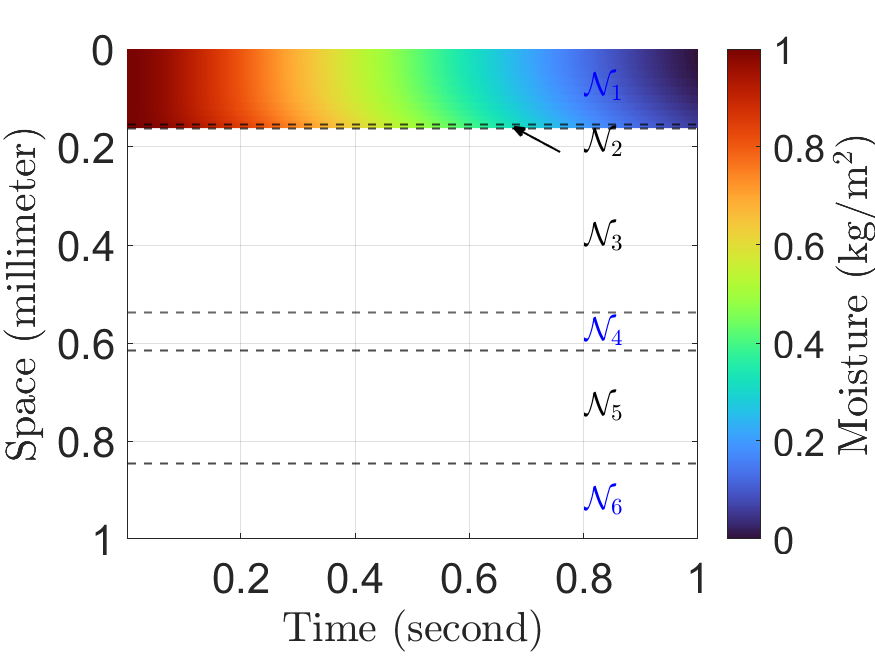}
        \caption{ Moisture variation over space and time for all the nodes.}
        \label{115moist3d}
    \end{subfigure}
    \caption{Simulation results of 115 $\mathrm{g/m^2}$ paper printed with 10 $\mathrm{g/m^2}$ of ink and conveyor. In \subref{115moist3d}, (\protect\blacklinedashed) represents the boundary for each node. The black arrow shows the node with ink and moisture. The plot is blank for nodes that do not have a moisture state.}
    \label{subfigure 115gsm moist}
\end{figure}

\section{Digital Twining of Fixation Process through $\mathcal{H}_{\infty}$ Optimal State Estimator Synthesis}\label{estimator section}
Beyond simulation and rapid-prototyping, the model developed in section 3 alone is not sufficient to monitor the actual states of the fixation process during printing operation. In this paper, we use the derived model as the basis for establishing a state estimator that acts as an adaptive digital twin to monitor the evolution of thermo-fluidic behavior in a fixation unit. At the same time, this state estimator must also be robust towards unknown perturbations that may corrupt the dynamics and provide reliable estimation despite them. 
\begin{figure}[H]
    \centering
    \includegraphics[scale=0.320]{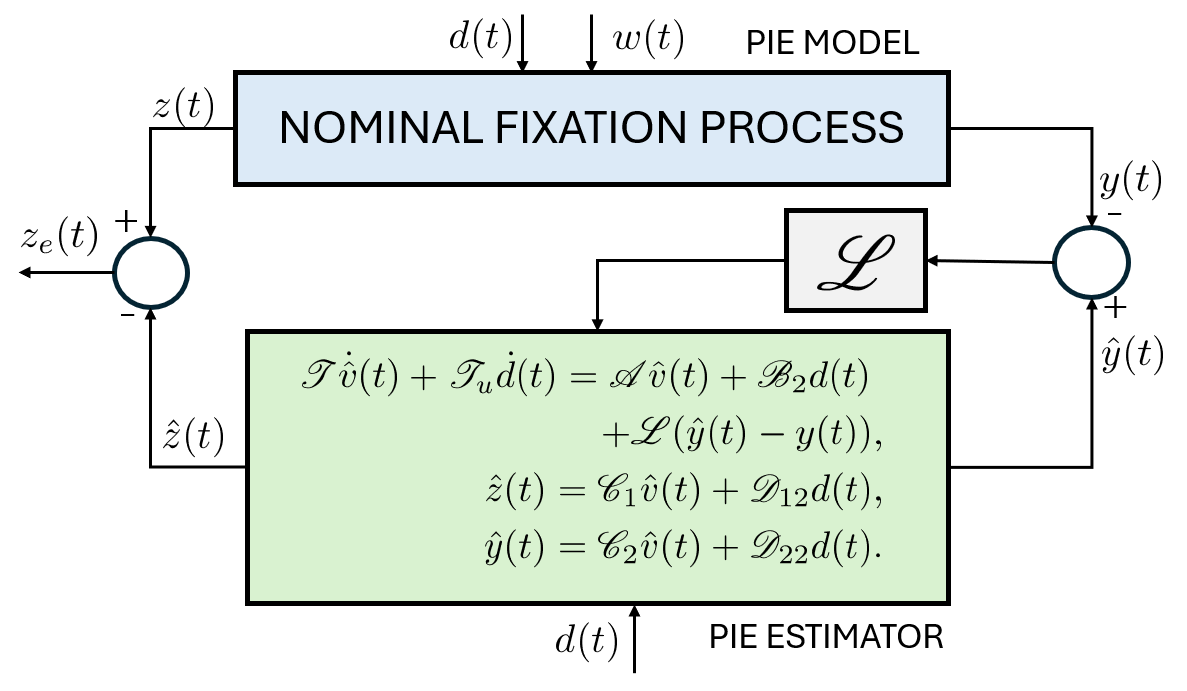}
    \caption{The Luenberger state estimator. $w(t)$ is the unknown exogenous disturbance.}
    \label{Lestimator}
\end{figure} 
In this paper, taking advantage of the computability of the PIE representation, we propose a Luenberger state estimator as shown in Figure \ref{Lestimator}. The adaptability of the digital twin results from the introduction of the estimator gain $\mathcal{L}$, which explicitly takes into account the difference between the sensor measurements and the predicted output in estimation.

For simplicity and ease of testing, we first present the results by not considering the non-linear function $\Delta$ for the estimator synthesis, thus substituting $\Delta= \varnothing$. Thus, for the nominal fixation process, the dynamic equations and the output channels of the true system are given by
\begin{equation}\label{truesys}
    \begin{split}
          \mathscr{T}\dot{v}(t) + \mathscr{T}_w \dot{w}(t)+ \mathscr{T}_u \dot{d}(t)= \mathscr{A}v(t) +\mathscr{B}_1 w(t) \\+ \mathscr{B}_2 d(t),  \\ 
            z(t) = \mathscr{C}_1v(t)+\mathscr{D}_{11}w(t)+\mathscr{D}_{12}d(t),\\    y(t)=\mathscr{C}_2v(t)+\mathscr{D}_{21}w(t)+\mathscr{D}_{22}d(t). 
    \end{split}
\end{equation}
The dynamic equations and the output channels of the estimator are given by
\begin{equation}\label{estsys}
    \begin{split}
          \mathscr{T}\dot{\hat{v}}(t) + \mathscr{T}_u \dot{d}(t)= \mathscr{A}\hat{v}(t) + \mathscr{B}_2 d(t) \ \\+\mathscr{L}(\hat{y}(t)-y(t)), \\ 
  \hat{z}(t) = \mathscr{C}_1\hat{v}(t)+\mathscr{D}_{12}d(t),\\ 
\hat{y}(t) = \mathscr{C}_2\hat{v}(t)+\mathscr{D}_{22}d(t). 
    \end{split}
\end{equation}
The error system and the estimation error signal obtained by solving \eqref{truesys} and \eqref{estsys}, are given by 
\begin{equation}
\begin{split}
    \mathscr{T}\dot{e}(t) + \mathscr{T}_w \dot{w}(t) = &(\mathscr{A}+\mathscr{L}\mathscr{C}_2)e(t)\\+&(\mathscr{B_1}+\mathscr{L}\mathscr{D}_{21})w(t)\\
    z(t)-\hat{z}(t) = &\mathscr{C}_1e(t)+\mathscr{D}_{11}w(t)
\end{split}
\end{equation}
The objective is to synthesize an estimator gain $\mathscr{L}$, and the smallest value of $\mathcal{H}_{\infty}$ gain $\gamma$ such that $||z_e(t)||_{2} \leq \gamma ||w||_{2}(t)$, and $z_e(t) = z(t)-\hat{z}(t)$. Here, $||z_e(t)||_2$ and $||w(t)||_2$ are the $L_2$ norm of the estimation error and disturbance signals respectively (assuming these signals are square integrable).\\
\begin{prop} Suppose there exists scalars $\epsilon,\gamma >0$. Now consider the following optimization problem:
\[
\hat{\gamma} = \arg\min \gamma
\]
subject to
\begin{itemize}
    \item $\mathscr{P}:=\fourpi{P}{\mbf{Q_1}}{\mbf{Q_2}}{\mbf{R_i}} \geq \epsilon I$,
    \item \[\begin{bmatrix} \mathscr{T}_w^* (\mathscr{P} \mathscr{B}_1 + \mathscr{Z} \mathscr{D}_{21}) + (\cdot)^* & 0 & (\cdot)^* \\ 0 & 0 & 0 \\ -(\mathscr{P} \mathscr{A} + \mathscr{Z} \mathscr{C}_2)^* \mathscr{T}_w & 0 & 0 \end{bmatrix}\]\[ + \begin{bmatrix} -\gamma I & -\mathscr{D}_{11}^{*} & -(\mathscr{P} \mathscr{B}_1 + \mathscr{Z} \mathscr{D}_{21})^* \\ (\cdot)^* & -\gamma I & \mathscr{C}_1 \\ (\cdot)^* & (\mathscr{P} \mathscr{A} + \mathscr{Z} \mathscr{C}_2)^* & (\cdot)^* \end{bmatrix} \leq 0.
\]
\end{itemize}
    Then $\mathscr{P}^{-1}$ exists and is a bounded linear operator. Furthermore, $\mathscr{L}=\mathscr{P}^{-1} \mathscr{Z}$, for any $0 \neq w, z_e \in L_2 [0,\infty)$, satisfies 
    \[\sup \frac{||z_e||_{2}}{||w||_{2}} \leq \hat{\gamma}\]
\end{prop}

\begin{proof}
    The proof follows from \cite{phdthesis}, Theorem 6.7.
\end{proof}
\textcolor{black}{The inequalities in Proposition 3 are operator inequalities where all the operators are PI operators and the unknown PI operators ($\mathscr{P}, \mathscr{Z}$) appear linearly in the inequalities. In such a case, these operator inequalities (as well as a convex optimization problem subject to such inequalities) can be solved using Linear Matrix Inequalities (LMIs). Formulation of these LMIs, solving them, and inversion of $\mathscr{P}$ are performed in the \texttt{PIETOOLS} software package.}
\subsection{$\mathcal{H}_{\infty}$ optimal estimator for blank paper}
The measured output is given by $y(t) = \mathrm{col}(\mathbf{x}_{1,1}(s_{\iota,1},t),\\\mathbf{x}_{6,1}(s_{\upsilon,6},t))$. The regulated output is chosen to be $z_3(t) = \mathbf{x}_{3,1}(s_{\upsilon,3},t)$. The exogenous disturbance is chosen to be $w_1(t)=10 e^{-0.1t}sin(t)$, and the system is simulated for 50 seconds. Table \ref{esterror} shows the different initial conditions considered for the estimator and the maximum error in the estimation.
\begin{table}[H]
\centering
\caption{Maximum estimation error for different initial condition choices.}
\begin{tabular}{c c c} 
  Initial condition & Error signal & \makecell{peak error \\($^\circ$C)} \\ [0.5ex] 
 \hline
  Same as the true system & $z_{e,1}(t)$ & 0.3488  \\
 5\% error from the true system & $z_{e,2}(t)$ & 2.4525  \\
   10\% error from the true system & $z_{e,3}(t)$ & 4.9053  \\
    Unknown (zero) & $z_{e,4}(t)$ & 49.0567  \\[1ex] 
 \hline
\end{tabular}
\label{esterror}
\end{table}
The $\mathcal{H}_{\infty}$ gain $\gamma$ was equal to 0.08.
 Figure \ref{estimatorall} shows the estimation error over time plotted against the disturbance.
\begin{figure}[H]
    \centering
    \includegraphics[scale=0.5]{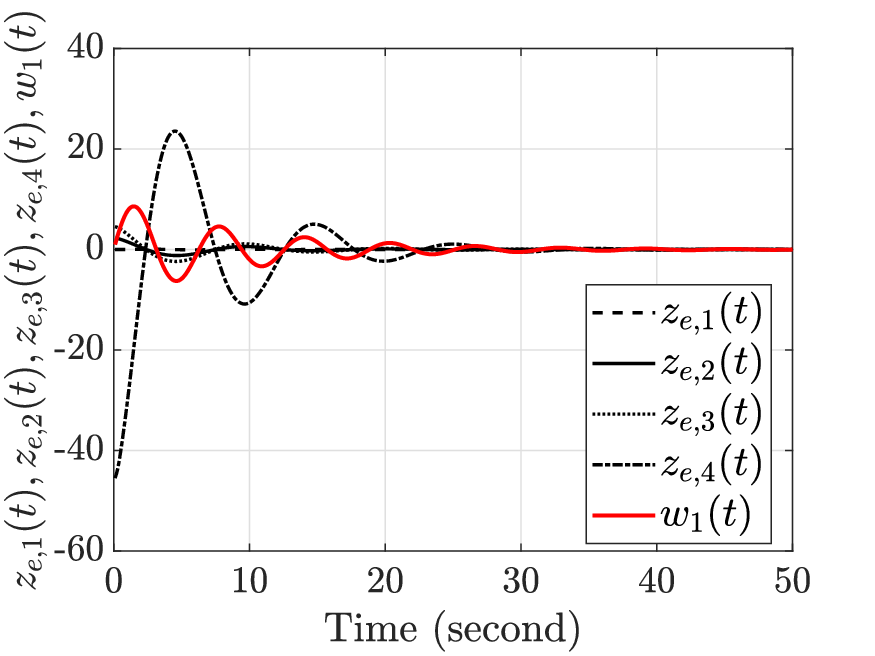}
    \caption{ Disturbance and estimation error for different initial conditions.}
    \label{estimatorall}
\end{figure}
\begin{figure*}[t]
    \centering
    \begin{subfigure}[b]{0.5\columnwidth}
        \centering
        \includegraphics[width=\linewidth]{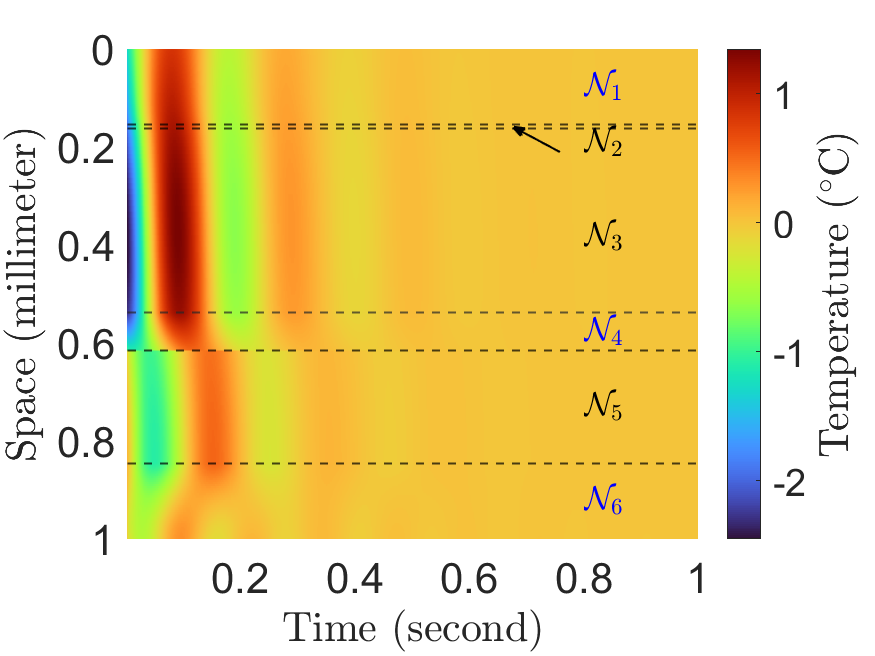}
        \caption{ 5\% off in initial condition.}
        \label{est5}
    \end{subfigure}
    \begin{subfigure}[b]{0.5\columnwidth}
        \centering
        \includegraphics[width=\linewidth]{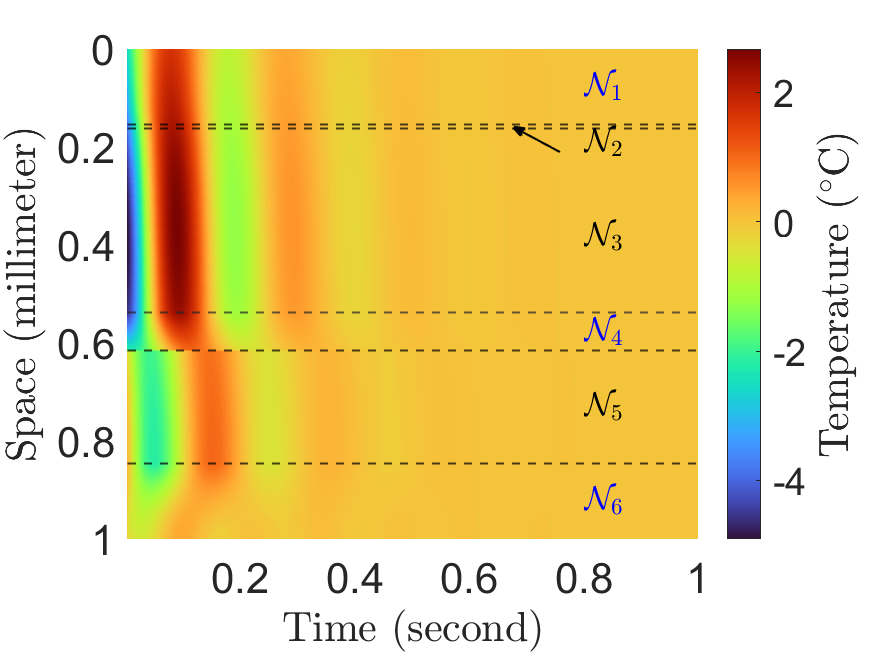} 
        \caption{ 10\% off in initial condition.}
        \label{est10}
    \end{subfigure}
    \begin{subfigure}[b]{0.5\columnwidth}
        \centering
        \includegraphics[width=\linewidth]{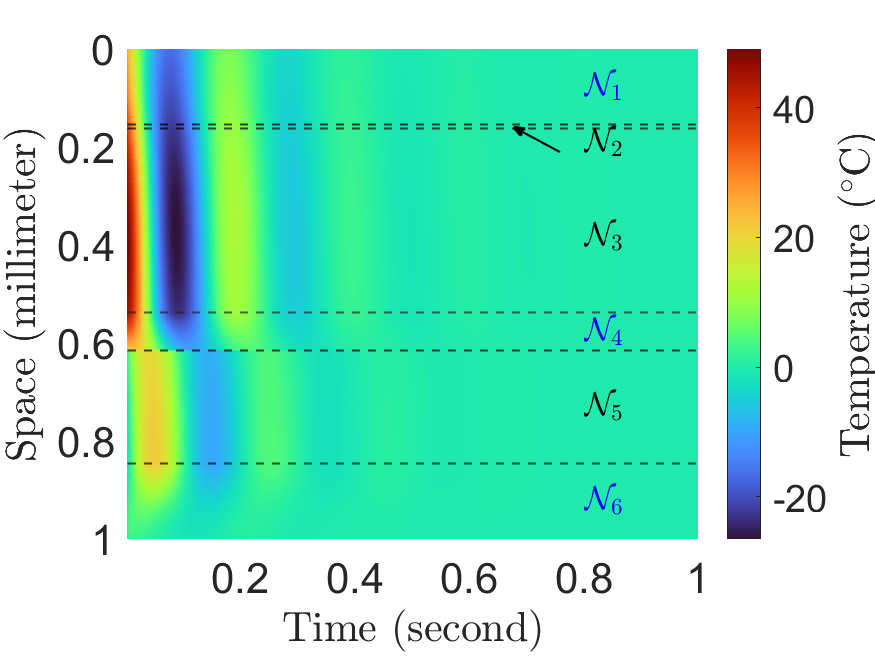} 
        \caption{ Zero initial condition.}
        \label{est0}
    \end{subfigure}
    \begin{subfigure}[b]{0.5\columnwidth}
        \centering
        \includegraphics[width=\linewidth]{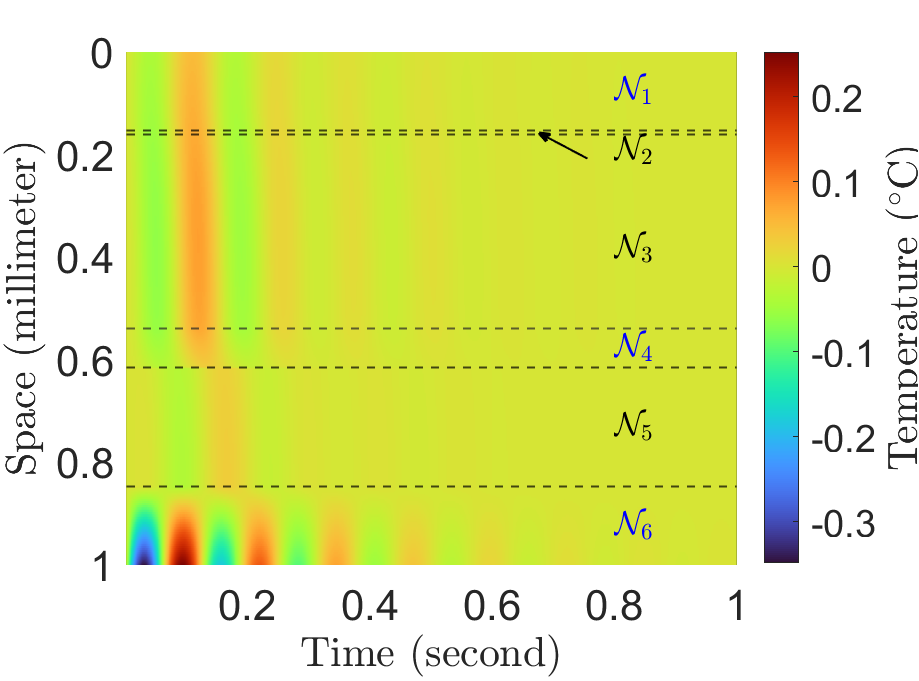}
        \caption{ Same initial condition.}
        \label{estsameic}
    \end{subfigure}
    \caption{The spatio-temporal variation of the error system with various offsets in the initial condition of the estimator with respect to the system and the maximum errors.}
    \label{fig:thrust_results}
\end{figure*}
Figures \ref{estsameic}, \ref{est5}, \ref{est10}, and \ref{est0} show the plots of the error system over the entire spatial domain simulated for 50 seconds. The space axis of all these plots is normalised in the range [0, 1]. \\
The maximum error occurs when the initial condition of the estimator is known in $\mathcal{N}_{6}$, while for uncertain initial conditions, the maximum error is observed in $\mathcal{N}_{3}$. This could be because, for the estimator with the known initial condition, the effect of the disturbance on the node $\mathcal{N}_{6}$ is felt with a delay because the heat has to diffuse through the spatial domain of all nodes and therefore would take time to attenuate it. However, when the initial condition of the estimator is uncertain, the estimator has to compensate for the mismatch in the initial condition along with the attenuation of the disturbance. Hence, this could be the reason why the maximum error occurs in $\mathcal{N}_{3}$.


\begin{rem}
The developed state estimator is not directly applicable for the PIE model \eqref{def_sol_funda}. In fact, the function $\Delta$ is quite complex and the development of a robust or parameter-varying state estimator to explicitly include nonlinearities remains an open problem and task for future research. On the other hand, a conservative robust state estimator can be synthesized by assuming no underlying structure of $\Delta$, i.e.,  by ignoring the relation $q = \Delta(p)$ and \eqref{def_sol_funda} becomes
\begin{align}
     \mathscr{T} \dot{v} +  \mathscr{T}_w \dot{w} +  \mathscr{T}_d \dot{d} &=\mathscr{A} v +  \mathscr{B}_1 w +  \mathscr{B}_2 d+ \mathscr{B}_q q,\notag\\
     \bmat{z\\y\\p} &= \bmat{\mathscr{C}_1&\mathscr{D}_{11}&\mathscr{D}_{12}&\mathscr{D}_{1q}\\\mathscr{C}_2&\mathscr{D}_{21}&\mathscr{D}_{22}&\mathscr{D}_{2q}\\\mathscr{C}_p&\mathscr{D}_{p1}&\mathscr{D}_{p2}&\mathscr{D}_{pq}} \bmat{v\\w\\d\\q}.
\end{align}
Thus, the above PIE system would represent every possible uncertainty and non-linearity. Now, a similar $\mathcal{H}_{\infty}$ state estimator synthesis can be repeated by considering additional channels $(p, q)$ in addition to the channels $(w, z_e)$. However, the performance of this new observer will be conservative and will not meet the tightness requirement of the industry.
\end{rem}
\section{Conclusions}
In this study, a modular digital twin is created that can effectively simulate the thermo-fluidic process in the fixation process and is compared with the sensor readings of the fixation with desirable prediction accuracy. Furthermore, an $\mathcal{H}_\infty$ state estimator is synthesized for the infinite-dimensional thermo-fluidic process, which can efficiently estimate the required thermal states, with a guaranteed robustness for worst-case external perturbation. This synthesis process is computable thanks to the Partial Integral Equation framework which provides quantified certificate of robustness of digital twin's performance during nominal printing operations. 

The current model considers evaporation as the primary non-linearity and condensation during fixation is not included, which can be further incorporated within this framework. Neumann boundaries cannot be imposed on both spatial ends due to the invertibility constraint of the $B_T$
matrix, but enabling this could simplify the model. The current digital twin supports only one-dimensional diffusion–transport–reaction; extending it to two dimensions would enhance accuracy, especially for belts with high thermal conductivity.

\nobalance
\bibliography{reference}             

@phdthesis{phdthesis,
title = "A digital twin for controlling thermo-fluidic processes",
author = "Amritam Das",
note = "Proefschrift.",
year = "2020",
month = nov,
day = "2",
language = "English",
isbn = "978-90-386-5140-8",
publisher = "Technische Universiteit Eindhoven",
type = "Phd Thesis 1 (Research TU/e / Graduation TU/e)",
school = "Electrical Engineering",
}

@INPROCEEDINGS{9029595,
  author={Das, Amritam and Shivakumar, Sachin and Weiland, Siep and Peet, Matthew M.},
  booktitle={2019 IEEE 58th Conference on Decision and Control (CDC)}, 
  title={$\mathcal{H}_{\infty}$  Optimal Estimation for Linear Coupled PDE Systems}, 
  year={2019},
  volume={},
  number={},
  pages={262-267},
  doi={10.1109/CDC40024.2019.9029595}}

@misc{pietools,
  author = {S. Shivakumar, M. M. Peet},
  title = {PIETOOLS},
  url = {control.asu.edu/pietools},
  version = {2020a},
}

@INPROCEEDINGS{9147712,
  author={Shivakumar, Sachin and Das, Amritam and Peet, Matthew M.},
  booktitle={2020 American Control Conference (ACC)}, 
  title={PIETOOLS: A Matlab Toolbox for Manipulation and Optimization of Partial Integral Operators}, 
  year={2020},
  volume={},
  number={},
  pages={2667-2672},
  doi={10.23919/ACC45564.2020.9147712}}

@INPROCEEDINGS{das_2019CDC,
  author =       {A. Das and S. Shivakumar and S. Weiland and M. Peet},
  title =        {{${H}_\infty$} Optimal Estimation for Linear Coupled {PDE} Systems},
  booktitle =    CDC,
  year =         {2019},
}

@book{crompton2003printing,
  title={The Printing Press},
  author={Crompton, S.W.},
  isbn={9780791074510},
  lccn={2003014059},
  series={Transforming power of technology},
  year={2003},
  publisher={Chelsea House Publishers}
}

@article{10.1145/1064830.1064860,
author = {S\'{e}quin, Carlo H.},
title = {Rapid prototyping: a 3d visualization tool takes on sculpture and mathematical forms},
year = {2005},
issue_date = {June 2005},
publisher = {Association for Computing Machinery},
address = {New York, NY, USA},
volume = {48},
number = {6},
issn = {0001-0782},

doi = {10.1145/1064830.1064860},
journal = {Commun. ACM},
month = {jun},
pages = {66–73},
numpages = {8}
}

@book{holik2013handbook,
  title={Handbook of Paper and Board},
  author={Holik, H.},
  isbn={9783527652518},
  year={2013},
  publisher={Wiley}
}

@book{pletcher2012computational,
  title={Computational Fluid Mechanics and Heat Transfer, Third Edition},
  author={Pletcher, R.H. and Tannehill, J.C. and Anderson, D.},
  isbn={9781591690375},
  lccn={2012028504},
  series={Series in Computational and Physical Processes in Mechanics and Thermal Sciences},
  year={2012},
  publisher={Taylor \& Francis}
}

@article{PEET2024115673,
title = {A new treatment of boundary conditions in PDE solution with Galerkin methods via Partial Integral Equation framework},
journal = {Journal of Computational and Applied Mathematics},
volume = {442},
pages = {115673},
year = {2024},
issn = {0377-0427},
author = {Yulia T. Peet and Matthew M. Peet}
}

@article{SHANG2000533,
title = {Feedback Control of Hyperbolic PDE Systems},
journal = {IFAC Proceedings Volumes},
volume = {33},
number = {10},
pages = {533-538},
year = {2000},
note = {IFAC Symposium on Advanced Control of Chemical Processes 2000, Pisa, Italy, 14-16 June 2000},
issn = {1474-6670},
author = {Huilan Shang and J. {Fraser Forbes} and Martin Guay}
}

@InProceedings{10.1007/978-3-0348-8849-3_12,
author="Ito, K.
and Ravindran, S. S.",
editor="Desch, W.
and Kappel, F.
and Kunisch, K.",
title="A Reduced Basis Method for Control Problems Governed by PDEs",
booktitle="Control and Estimation of Distributed Parameter Systems",
year="1998",
publisher="Birkh{\"a}user Basel",
address="Basel",
pages="153--168",
isbn="978-3-0348-8849-3"
}

@misc{BusinessResearchInsights2024,
  author       = {Business Research Insights},
  title        = {Global Production Printer Market Report 2024–2033},
  year         = {2024},
}

@ARTICLE{9465747,
  author={Das, Amritam and Princen, Martijn and Roudbari, Mahnaz Shokrpour and Khalate, Amol and Weiland, Siep},
  journal={IEEE Transactions on Control Systems Technology}, 
  title={Soft Sensing-Based In Situ Control of Thermofluidic Processes in DoD Inkjet Printing}, 
  year={2022},
  volume={30},
  number={3},
  pages={956-971},
  doi={10.1109/TCST.2021.3087576}}

@INPROCEEDINGS{2024arXiv241101793B,
  author={Braghini, Danilo and Shivakumar, Sachin and Peet, Matthew M.},
  booktitle={2025 IEEE 64th Conference on Decision and Control (CDC)}, 
  title={H2-Optimal Estimation of Linear Delayed and PDE Systems}, 
  year={2025},
  volume={},
  number={},
  pages={1422-1427},
  doi={10.1109/CDC57313.2025.11312977}}

@article{sachin,
  author={Shivakumar, Sachin and Das, Amritam and Weiland, Siep and Peet, Matthew},
  journal={IEEE Transactions on Automatic Control}, 
  title={Extension of the Partial Integral Equation Representation to GPDE Input–Output Systems}, 
  year={2025},
  volume={70},
  number={5},
  pages={3240-3255},
  doi={10.1109/TAC.2024.3505954}}

@article{PEET2021109473,
title = {A Partial Integral Equation (PIE) representation of coupled linear PDEs and scalable stability analysis using LMIs},
journal = {Automatica},
volume = {125},
pages = {109473},
year = {2021},
issn = {0005-1098},
doi = {https://doi.org/10.1016/j.automatica.2020.109473},
author = {Matthew M. Peet}
}

@INPROCEEDINGS{9030224,
  author={Shivakumar, Sachin and Das, Amritam and Weiland, Siep and Peet, Matthew M.},
  booktitle={2019 IEEE 58th Conference on Decision and Control (CDC)}, 
  title={A Generalized LMI Formulation for Input-Output Analysis of Linear Systems of ODEs Coupled with PDEs}, 
  year={2019},
  volume={},
  number={},
  pages={280-285},
  doi={10.1109/CDC40024.2019.9030224}}

@ARTICLE{11404177,
  author={Shivakumar, Sachin and Das, Amritam and Peet, Matthew},
  journal={IEEE Transactions on Automatic Control}, 
  title={Dual Representations and $H_{\infty }$-Optimal Control of Partial Differential Equations}, 
  year={2026},
  volume={},
  number={},
  pages={1-16},
  doi={10.1109/TAC.2026.3666695}}

@article{PEET2019132,
title = {Discussion Paper: A New Mathematical Framework for Representation and Analysis of Coupled PDEs},
journal = {IFAC-PapersOnLine},
volume = {52},
number = {2},
pages = {132-137},
year = {2019},
note = {3rd IFAC Workshop on Control of Systems Governed by Partial Differential Equations CPDE 2019},
issn = {2405-8963},
doi = {https://doi.org/10.1016/j.ifacol.2019.08.023},
author = {Matthew M. Peet and Sachin Shivakumar and Amritam Das and Seip Weiland}
}
                                                   







\end{document}